\documentclass[11pt]{article}

\usepackage{amsmath,amsfonts,amssymb,latexsym,epsfig}

\usepackage{fullpage}
\usepackage{amsthm}    
\usepackage{mathabx}    
\usepackage{mathtools}
\usepackage{graphicx}   
\usepackage{verbatim}   
\usepackage{color}      
\usepackage{caption}
\usepackage{epstopdf}
\usepackage{subcaption}
\usepackage{bm}

\newtheorem{theorem}{Theorem}[section]

\newtheorem{prop}[theorem]{Proposition}

\newcommand{\E}{{\mathbb E}}

\newcommand{\R}{{\mathbb R  }}

\newcommand{\erf}{\operatorname{erf}}

\begin{document}
\title{Dynamics of the Desai-Zwanzig model in multi-well and random energy landscapes}

\author{Susana N. Gomes 
\thanks{Mathematics Institute, University of Warwick (susana.gomes@warwick.ac.uk)}\and Serafim Kalliadasis 
\thanks{Department of Chemical Engineering, Imperial College London} \and
Grigorios A. Pavliotis
\thanks{Department of Mathematics, Imperial College London} \and
Petr Yatsyshin
\thanks{Department of Chemical Engineering, Imperial College London}}

\date{\today}

\maketitle

\begin{abstract}
We analyze a variant of the Desai-Zwanzig model [J. Stat. Phys. {\bf 19}
1-24 (1978)]. In particular, we study stationary states of the mean field
limit for a system of weakly interacting diffusions moving in a multi-well
potential energy landscape, coupled via a Curie-Weiss type (quadratic)
interaction potential. The location and depth of the local minima of the
potential are either deterministic or random. We characterize the structure
and nature of bifurcations and phase transitions for this system, by means
of extensive numerical simulations and of analytical calculations for an
explicitly solvable model. Our numerical experiments are based on Monte
Carlo simulations, the numerical solution of the time-dependent nonlinear
Fokker-Planck (McKean-Vlasov equation), the minimization of the free energy
functional and a continuation algorithm for the stationary solutions.
\end{abstract}

%

\section{Introduction}\label{sec:intro}
%

Systems of interacting particles, often subject to thermal noise, arise in a
wide spectrum of natural phenomena and applications, ranging from plasma
physics and galactic dynamics~\cite{BinneyTremaine2008} to dynamical
density-functional theory (DDFT)~\cite{Ben_2012a,Ben_2012b}, mathematical
biology~\cite{Farkhooi2017,Lucon2016} and even in mathematical models in
social sciences~\cite{GPY2017,Motsch2014}.  As examples of models of
interacting ``agents" in a noisy environment that appear in the social
sciences, including crowd dynamics, we mention the modeling of cooperative
behavior~\cite{Dawson1983}, risk management~\cite{GPY2012} and opinion
formation~\cite{GPY2017}. Other recent applications that have motivated this
work are global optimization~\cite{Pinnau_al2017}, active
media~\cite{Bain-Bartolo_2017} and machine learning. Indeed, it has been
shown recently~\cite{rotskoff_vanden-eijnden2018, SirignanoSpiliopoulos2018}
that ``stochastic gradient descent", the optimization algorithm used in the
training of neural networks, can be represented  as the evolution of a
particle system with interactions governed by a potential related to the
objective function that is used to train the network. Several of the issues
that we study here, such as phase transitions and the effect of nonconvexity,
are of great interest in the context of the training of neural networks.

For weakly interacting diffusions, one can pass rigorously to the mean-field
limit leading to the McKean-Vlasov equation, a nonlinear nonlocal
Fokker-Planck type Eq.~\cite{mckean, Dawson1983}. Unlike finite systems
of interacting diffusions, whose law (probability density function) is
governed by the linear Fokker-Planck equation, the McKean-Vlasov equation can
exhibit phase transitions~\cite{Dawson1983}. Indeed, whereas a finite system
of interacting overdamped Langevin diffusions moving in a confining potential
is always ergodic with respect to the Gibbs measure~\cite[Ch. 4]{Pavl2014},
the McKean-Vlasov equation with a non-convex confining potential, can have
several stationary solutions at low temperatures~\cite{Dawson1983,
Tamura1984}. As a matter of fact, the number of stationary solutions depends
on the number of metastable states (local minima) of the confining
potential~\cite{Tugaut2014}. A complete rigorous analysis of phase
transitions, both continuous and discontinuous, for the McKean-Vlasov
dynamics in a box with periodic boundary conditions and for non-convex (i.e.
non-H stable) interaction potentials is presented
in~\cite{Pavliotis_al_2018}. The mean field limit for non-Markovian
interacting particles, including the effect of memory on the bifurcation
diagram, is studied in~\cite{DuongPavliotis2018}.

The main purpose of this study is to scrutinize the dynamics of a system of
weakly interacting diffusions and, in particular, characterize bifurcations
and phase transitions for this system in the presence of a multi-well
confining potential which can have random locations and depths of local
minima, interacting under a quadratic Curie-Weiss potential. An example of a
deterministic multi-well potential is given in Fig.~\ref{Fig:potentials}. It
is a modified version of the so-called M$\ddot{\rm u}$ller-Brown
potential~\cite{Muller-Brown_1979}, a canonical potential surface used often
as a prototype in theoretical chemistry including reaction
dynamics~\cite{Kawai-Komatsuzaki_2010}, but also theoretical biology
including protein folding~\cite{Chekmarev_2015}. The potential is also often
adopted as a prototype to test the performance of computational optimization
algorithms to e.g. obtain reaction paths~\cite{Bonfanti-Kob_2017}. Multi-well
potentials/rugged energy landscapes have numerous applications, from
materials science and catalysis where (surface) diffusion in a multiscale
potential is critical to understanding how atoms or molecules adsorb on
catalytic surfaces and react, to droplet motion on chemically heterogeneous
substrates~\cite{Keil_2012}.
\begin{figure}
\includegraphics[width=0.4\textwidth]{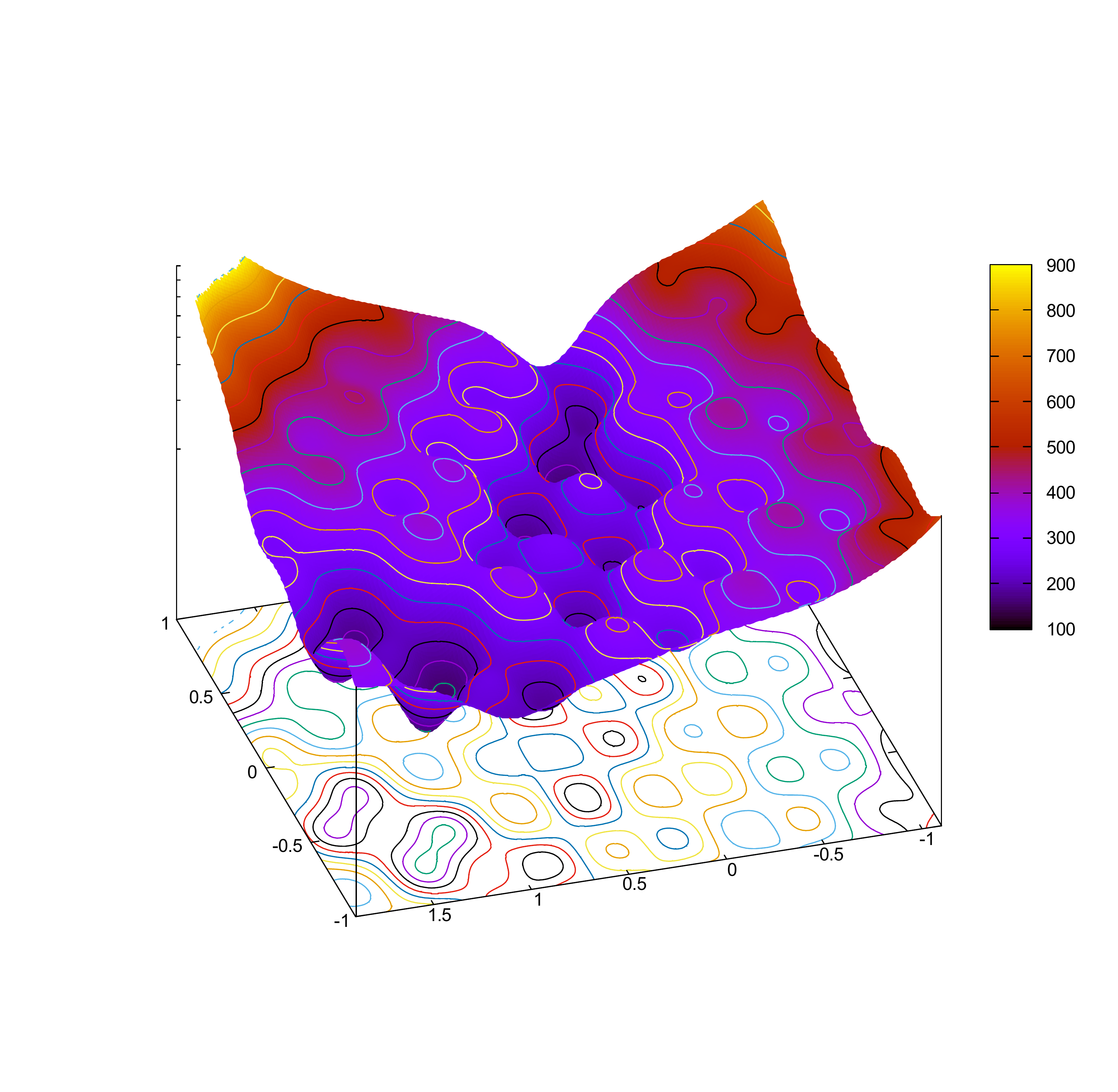}
\centering
\caption{A $2D$ multiscale potential.}
\label{Fig:potentials}
\end{figure}
Our study builds upon earlier work~\cite{GomesPavliotis2017}. There, the mean
field limit for interacting diffusions in a two-scale, locally periodic
potential was considered. This problem was then studied using tools from
multiscale analysis, and in particular periodic homogenization for parabolic
PDEs~\cite{PavlSt08}. In contrast, here the focus is on multi-well potentials
(either deterministic or random) that do not have a periodic structure, and,
consequently, the theory of periodic homogenization is not applicable.

In particular, we will offer a complete bifurcation analysis and explicit
characterization of phase transitions for the McKean-Vlasov equation in one
dimension for model multi-well potentials with an arbitrary number of local
minima. And we will also study phase transitions when the number of local
minima tends to infinity.

Our starting point is a system of interacting particles in one dimension,
moving in a confining potential, $V(\cdot)$, and which interact through an
interaction potential $W(\cdot)$ which we consider to be of a quadratic
Curie-Weiss type (i.e., $W(x) = \frac{x^2}{2}$):
\begin{equation}\label{e:sde-general} d X^i_t =
\left(-V'(X_t^i) - \frac{\theta}{N} \sum_{j=1}^N (X_t^i-X_t^j)\right) \, dt +
\sqrt{2 \beta^{-1}} \, dB_t^i,
\end{equation}
for $i=1,\dots,N$. Here $\{  X_t^i \}_{i=1}^N$ denotes the position of the
interacting agents, $\theta$ the strength of the interaction between the
agents, $\{  B_t^i \}_{i=1}^N$ standard independent one-dimensional Brownian
motions and $\beta$ denotes the inverse temperature. The total energy
(Hamiltonian) of the system of interacting diffusions~\eqref{e:sde-general}
is
\begin{equation}\label{e:energy}
W_N({\bf X}) = \sum_{\ell =1}^N V\left(X^{\ell}\right) + \frac{\theta}{4 N} \sum_{n=1}^N \sum_{\ell =1}^N \left(X^n - X^{\ell}\right)^2,
\end{equation}
where ${\bf X}=(X^1,\dots,X^N)$.

Passing rigorously to the mean field limit as $N\rightarrow\infty$ using, for example, martingale techniques~\cite{Dawson1983,Gartner1988,Oelschlager1984},
and under appropriate assumptions on the confining potential and on the initial conditions (propagation of chaos), is a well-studied problem.
Formally, using the law of large numbers we deduce that
$$
\lim_{N \rightarrow +\infty} \frac{1}{N} \sum_{j=1}^N X^j_t = \E X_t,
$$
where the expectation is taken with respect to the ``1-particle" distribution
function $p(x,t)$ -- This corresponds to the mean field ansatz for the
$N-$particle distribution function, $p_N(x_1, \dots x_N,t ) = \prod_{n=1}^N
p(x_n,t)$ and taking the limit as $N\rightarrow \infty$;
see~\cite{MartzelAslangul2001,balescu97}.--  Passing, formally, to the limit
as $N\rightarrow \infty$ in the stochastic differential
equation~\eqref{e:sde-general}, we obtain the McKean stochastic differential
equation (SDE)
\begin{equation}\label{e:mckean-sde}
d X_t = -V'(X_t) \, dt - \theta (X_t - \E X_t) \, dt + \sqrt{2 \beta^{-1}} \, dB_t.
\end{equation}
The Fokker-Planck equation corresponding to this SDE is the McKean-Vlasov equation~\cite{frank04,McKean1966,mckean}
\begin{equation}\label{e:mckean-vlasov}
\frac{\partial p}{\partial t} = \frac{\partial}{\partial x} \left(V'(x) p + \theta \left(W'\star p\right) p + \beta^{-1} \frac{\partial p}{\partial x} \right),
\end{equation}
where $\star$ denotes the convolution operator. The McKean-Vlasov equation is
a nonlinear, nonlocal equation, sometimes referred to as the
McKean-Vlasov-Fokker-Planck equation. It is a gradient flow, with respect to
the Wasserstein metric, for the free energy functional
\begin{equation}\label{e:free-energy}
\mathcal{F}[\rho] = \beta^{-1} \int \rho \ln \rho \, dx + \int V \rho \, dx + \frac{\theta}{2} \int \int W(x-y) \rho(x) \rho(y) \, dx \, dy.
\end{equation}
Background material on the McKean-Vlasov equation can be found in,
e.g.~\cite{frank04, CMV2006,Villani2003}. The equation belongs to the broad
class of (entropic) gradient flows. Various technical questions, such as the
connection between gradient flows and large deviations are addressed
in~\cite{Dirr-etal_2016,Adams-etal_2018}.

It is noteworthy that the system of interacting Langevin
equations~\eqref{e:sde-general} as well as the potential
energy~\eqref{e:energy} retain the basic features of the models studied in
DDFT for classical fluids~\cite{Ben_2012a,Ben_2012b}. One approach used to
derive DDFT is to start with the Langevin dynamics of Brownian particles to
obtain a Fokker-Planck equation for the $n$-particle probability
distribution. A formal BBGKY hierarchy then is used to obtain a closed
equation for the density distribution. The main assumption is an equilibrium
thermodynamic sum rule, the so-called adiabatic approximation, by which the
higher-body correlations are approximated by those of an equilibrium fluid
with the same density distribution. Including hydrodynamic interactions in
DDFT to obtain a hydrodynamic description that includes intermolecular
interactions is non-trivial~\cite{Ben_2012b,Ben_2013}.

On the other hand, the Vlasov equation, originally derived for ionized gases,
i.e. plasma with long-range (Coulomb) forces, describing the time evolution
of the distribution function of plasma, is somewhere halfway between an
$n$-particle model and the hydrodynamic description. The latter can be
rigorously obtained starting from the Langevin dynamics assuming there is no
bath, replace the potential with the Boltzman collision operator to get the
Boltzman equation, and then use homogenization (as in ~\cite{Ben_2012a}). But
the standard kinetic approach based on the Boltzman equation cannot be
applied for Coulomb forces which are long range and cannot be treated as
collisions (unless a gas is weakly ionized, in which case the charge-charge
interactions are negligible with respect to collisions with neutrals).
Instead, one has to start with the collisionless Boltzman equation and
appropriately adapt it to plasma by also utilizing Maxwell's equations (but
in most cases of practical interest that time variations of the fields in the
Maxwell's equations are negligible) and the quasi-static fields are
considered instead.

The finite dimensional dynamics~\eqref{e:sde-general} has a unique invariant measure.
Indeed, the process ${\bf X}_t$ defined in~\eqref{e:sde-general} is always
ergodic, and in fact reversible, with respect to the Gibbs measure~\cite[Ch. 4]{Pavl2014},
\begin{subequations}\label{e:gibbs-N}
\begin{eqnarray}
\mu_N (dx) &= \frac{1}{Z_N} e^{-\beta W_N(x^1, \dots x^N)} \, dx^1 \dots dx^N, \\
Z_N &= \int_{\R^N} e^{-\beta W_N(x^1, \dots x^N)} \, dx^1 \dots dx^N
\end{eqnarray}
\end{subequations}
where $W_N(\cdot)$ is given by~\eqref{e:energy}. On the other hand, the
McKean dynamics~\eqref{e:mckean-sde} and the corresponding
McKean-Vlasov-Fokker-Planck equation~\eqref{e:mckean-vlasov} can have more
than one invariant measures, for nonconvex confining potentials and at
sufficiently low temperatures~\cite{Dawson1983,Tamura1984}. This is not
surprising, since the McKean-Vlasov equation is a nonlinear, nonlocal PDE,
as already noted, and the standard uniqueness of solutions for the linear
(stationary) Fokker-Planck equation does not apply~\cite{BKRS2015}.

The density of the invariant measure(s) for the McKean dynamics~\eqref{e:mckean-sde} satisfies the stationary nonlinear Fokker-Planck equation
\begin{equation}\label{e:mckean-station}
\frac{\partial}{\partial x} \left(V'(x) p_{\infty} + \theta \left(W'\star p \right) p_{\infty} + \beta^{-1} \frac{\partial p_{\infty}}{\partial x} \right) =0.
\end{equation}
Based on earlier work~\cite{Dawson1983,Tamura1984}, it is by now well understood that the number of invariant measures, i.e. the number of solutions
to~\eqref{e:mckean-station}, is related to the number of metastable states (local minima) of  the confining potential -- see~\cite{Tugaut2014} and the references therein.

For the Curie-Weiss (i.e. quadratic) interaction potential, we can write
Eq.~\eqref{e:mckean-vlasov} as a Fokker-Planck equation with a dynamic
constraint
\begin{eqnarray}
\frac{\partial p}{\partial t} &=& \beta^{-1}\frac{\partial^2 p}{\partial x^2} +\frac{\partial}{\partial x} \left(V'(x)p - \theta \left(m-x\right) p \right),\\
m &=& R(m) = \int_\R xp(x,t) \ dx,
\end{eqnarray}
and, from the corresponding steady-state equation, a one-parameter family of
solutions to the stationary McKean-Vlasov equation~\eqref{e:mckean-station}
can be obtained:
\begin{subequations}\label{e:inv-meas-mckean}
\begin{eqnarray}
p_{\infty}(x ; \theta, \beta, m) &=& \frac{1}{Z(\theta, \beta ; m )} e^{- \beta \left( V(x) + \theta \left(\frac{1}{2}x^2 - x m \right) \right)}, \\
\label{e:inv-meas-mckean-b} Z(\theta, \beta ; m ) &=& \int_{\R} e^{- \beta \left( V(x) + \theta \left(\frac{1}{2}x^2 - x m \right) \right)} \, dx.
\end{eqnarray}
\end{subequations}
This one-parameter family of probability densities is subject, of course, to the constraint that it provides us with the correct formula for the first moment:
\begin{equation}\label{e:self-consist}
m = \int_{\R} x p_{\infty}(x ; \theta, \beta, m) \, dx =: R(m; \theta, \beta).
\end{equation}
We will refer to this as the {\bf selfconsistency} equation and it will be
the main object of study of this paper. Once a solution
to~\eqref{e:self-consist} has been obtained, substitution back
into~\eqref{e:inv-meas-mckean} yields a formula for the invariant density
$p_{\infty}(x ; \theta, \beta, m)$.

Clearly, the number of invariant measures of the McKean-Vlasov dynamics is
determined by the number of solutions to the selfconsistency
equation~\eqref{e:self-consist}. It is well known and not difficult to prove
that for symmetric nonconvex confining potentials a unique invariant measure
exists at sufficiently high temperatures, whereas more than one invariant
measures exist below a critical temperature $\beta^{-1}_c$~\cite[Thm.
3.3.2]{Dawson1983},~\cite[Thm. 4.1, Thm. 4.2]{Tamura1984}, see
also~\cite{shiino1987}. In particular, for symmetric potentials, $m = 0$ is
always a solution to the selfconsistency equation~\eqref{e:self-consist}.
Above $\beta_c$, i.e. at sufficiently low temperatures, the zero solution
loses stability and a new branch bifurcates from the $m = 0$
solution~\cite{shiino1987}. This second-order phase transition is similar to
the one familiar from the theory of magnetization and the study of the Ising
model. It will become clear later on this study that for multi-well
potentials the value of $\theta$ also plays a role on the type of
bifurcations obtained and therefore it is important to keep both parameters
in our analysis.

Another important property of the solutions of the Fokker-Planck equation~\eqref{e:mckean-vlasov} is the critical temperature $\beta_C$ at which pitchfork bifurcations from the mean
zero solution occur. This critical temperature is a function of $\theta$, and is given~\cite{shiino1987} by the solution to the equation
\begin{equation}\label{eq:beta_C}
\operatorname{Var}(m=0, \theta,\beta) := \int x^2 p_\infty(x;m=0,\theta,\beta) \ dx = \beta^{-1}\theta^{-1}.
\end{equation}
This equation can be solved numerically for the potentials we will study in
this paper.

The structure and number of equilibrium states for the generalized
Desai-Zwanzing model that we consider can be studied using four different
approaches:

\begin{enumerate}

\item As the invariant measure(s) of the particle dynamics~\eqref{e:sde-general}, in the limit as the number of particles becomes infinite.

\item As the long time limit of solutions to the time-dependent nonlinear Fokker-Planck equation~\eqref{e:mckean-vlasov}.

\item As minimizers of the free energy functional~\eqref{e:free-energy}.

\item In terms of solutions to the self-consistency equation~\eqref{e:self-consist}.

\end{enumerate}
We will use all of these in order to construct the bifurcation diagrams for
the stationary states of~\eqref{e:mckean-vlasov}.

The rest of the paper is organized as follows. In Section~\ref{s:models} we
briefly summarize the models (i.e., the types of confining potentials) we
consider in our study, and present the different methodologies we will use to
construct the bifurcation diagrams and analyze the stability of each branch.
In Section~\ref{sec:numeric} we present extensive numerical experiments we
performed to obtain the bifurcation diagram for different types of potentials,
including calculations of the free-energy surfaces associated with each system,
of the bifurcation diagrams of the first moment $m$ as a function of the inverse
temperature $\beta$ and of the critical temperature $\beta_C$ as a function of
$\theta$ as well as time-dependent simulations of the particle system and the
corresponding McKean-Vlasov equation. A discussion and conclusions are
offered in Section~\ref{s:conclusions}.

%
\section{Models studied and methodology}\label{s:models}
%

In this section we outline the model confining potentials that we will
consider and we also provide details of the mathematical and computational
techniques that we will use.

\subsection{Models studied in this paper}
\label{subsec:models}

Consider the system of weakly interacting diffusions given
in~\eqref{e:sde-general}. As already emphasized, the interaction is taken to
be of the Curie-Weiss type ($W(x) = \frac{x^2}{2}$) and different types of
confining multi-well potentials will be considered. In particular, we will
study the following potentials.
\begin{enumerate}
\item Polynomial potentials of the form (see, e.g., Fig.~\ref{poly_pot})
\begin{equation}\label{eq:poly_intro}
V(x) = \sum_{\ell=1}^{M} a_\ell x^{2m}.
\end{equation}
\begin{figure}[t!]
\centering
    \includegraphics[width=0.75\linewidth]{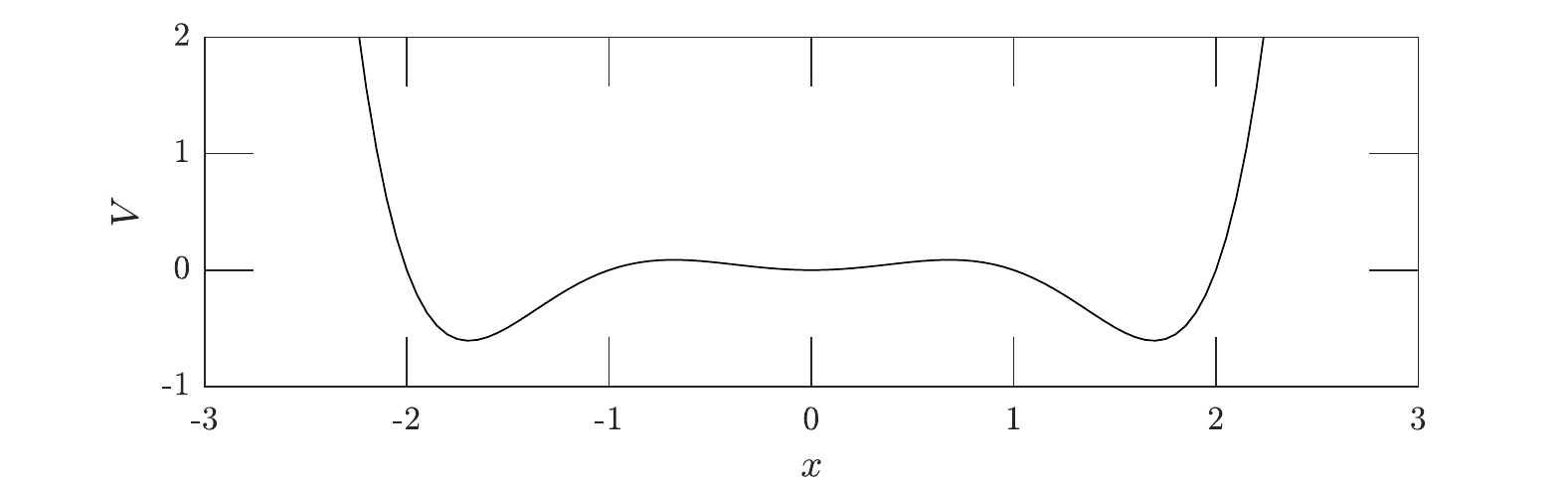}
    \caption{A polynomial potential of the form~\eqref{eq:poly_intro}, where $M=3$.
    \label{poly_pot}}
\end{figure}

\item Rational potentials (\cite{zofiaThesis}) with an arbitrary number of
    local minima and with (possibly) random location and depths of local
    minima, see, e.g., Fig.~\ref{fig:pot-random}.
\begin{equation}\label{eq:zofia_intro}
V(x) = \frac{1}{\sum_{\ell=-M}^M \delta_{\ell} |x - c_{\ell} x_{\ell} |^{-2}},
\end{equation}
where we consider both deterministic and random distributions of
$\{\delta_{\ell}, c_{\ell}\}$ and where in the random case we take these
distributions to be uniform.

\begin{figure}[h!]
\centering
\includegraphics[width=0.75\linewidth]{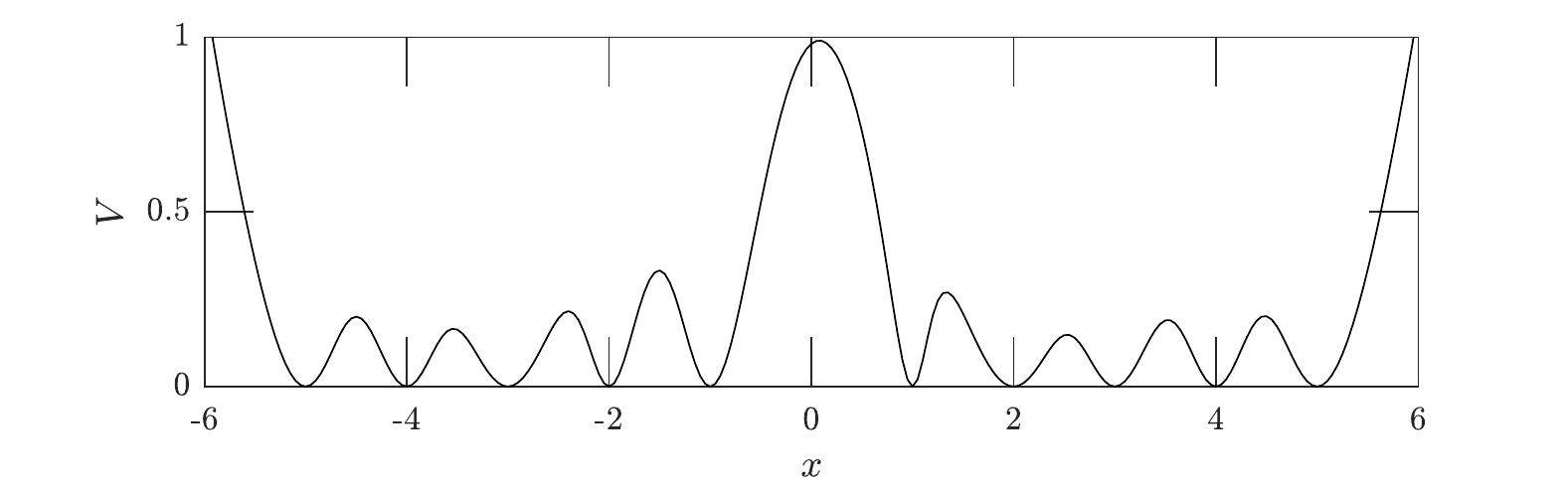}
\caption{A realization of a potential of the form~\eqref{eq:zofia_intro} with $10$ local minima located at the integers between $-5$ and $5$ ($c_\ell = 1$) and separated by arbitrary heights ($\delta_\ell\sim U([0,1])$).}\label{fig:pot-random}
\end{figure}

\item Piecewise linear potentials with quadratic growth at infinity,
\begin{equation}\label{eq:piecewise_intro}
V(x) = \left\{\begin{array}{lr}
\frac{x^2 -x_M^2}{2}, & \textrm{ if } |x|\geq x_M, \\
\frac{(H_i-h_i)x}{x_{i+\frac{1}{2}}-x_i} +\frac{h_ix_{i+\frac{1}{2}}-H_ix_i}{x_{i+\frac{1}{2}}-x_i},& \textrm{ if } x_i<x<x_{i+\frac{1}{2}}, \\
\frac{(h_i-H_i)x}{x_{i+1}-x_{i+\frac{1}{2}}} + \frac{H_ix_{i+1}-h_ix_{i+\frac{1}{2}}}{x_{i+1}-x_{i+\frac{1}{2}}},  & \textrm{ if } x_{i+\frac{1}{2}}<x<x_{i+1}.
\end{array}\right.
\end{equation}
\end{enumerate}

\subsection{Methodology}
\label{sub:method}

Our aim is to study the bifurcation diagram of the invariant measures of the
system of SDEs~\eqref{e:sde-general}. We will do that by considering the
mean-field limit of this system, given by Eq.~\eqref{e:mckean-sde}. In this
limit, the density of the particles satisfies the nonlinear nonlocal
Fokker-Planck equation~\eqref{e:mckean-vlasov}. The invariant measure(s) of
the system~\eqref{e:sde-general} satisfy the stationary Fokker-Planck
equation~\eqref{e:mckean-station}. Depending on the interaction potential,
and the parameters $\theta$ and $\beta$, there will exist only one (for
sufficiently small beta) or multiple (for large $\beta$) solutions. The
number of these solutions depends on the number of local minima and maxima of
the confining potential $V$.

We construct the bifurcation diagram using the first moment $m=m_1 = \mathbb{E}(X_t)$
as the order parameter and plot it as a function of the inverse temperature
$\beta$ for a fixed value of $\theta$. We will use two methods to obtain the
bifurcation diagram:
\begin{enumerate}
\item Using the fact that equilibrium states are minimizers of the
    free-energy functional~\eqref{e:free-energy}. We can then find the
    equilibrium points using differential geometry techniques. This
    methodology has the advantage that it immediately gives us the
    stability of each branch - stable solutions are local minima of the
    free-energy, while its local maxima are unstable solutions.

\item Arclength continuation to solve the self-consistency
    equation~\eqref{e:self-consist}. For the case of polynomial potentials,
    we can also use this technique to solve a system of ordinary
    differential equations (ODEs) for the moments (details are given in
    section~\ref{s:num_poly}). Our continuation scheme makes use of a
    modification of \textsc{Matlab}'s \textsf{matcont} routine (details are
    given in section \ref{Continuation}).
\end{enumerate}

Both methods have been used successfully in our previous studies to perform
detailed parametric studies and compute adsorption isotherms, bifurcations of
equilibrium states, phase diagrams and critical points for fluids in
confinement using DFT
(e.g.~\cite{Peter_2012,Peter_2013,Peter_2016,Peter_2018}).

Our results are confirmed by performing a careful comparison of the
bifurcation diagrams obtained with the long-time behavior of solutions to the
Fokker-Planck equation, as well as Monte-Carlo (MC) simulations of the
corresponding particle system.

\subsubsection{Free-energy formulation}
We make use the fact that the stationary solutions of the Fokker-Planck
equation~\eqref{e:mckean-vlasov} are equilibrium points of the free energy
given by Eq.~\eqref{e:free-energy}. In particular, since we know the form of
the steady solutions $p_\infty(m;\theta,\beta)$, we can compute the
free-energy surface as a function of the two arguments, $m$ and $\beta$, for
a fixed value of $\theta$, $\mathcal{F}(m,\beta;\theta)$. By computing its
equilibrium points, we can then plot the desired bifurcation diagram. This
methodology also allows us to immediately evaluate the stability of each
branch, since local minima correspond to stable solutions, while local maxima
are unstable ones. The bifurcation diagrams obtained with this method suffer
from poor resolution near branching points. However, the resolution at higher
values of $\beta$ allows us to easily find initial guesses for the arclength
continuation process which we describe below, and also guarantees that we
have information about all the existing branches.

A useful observation which we will make use of is that the free energy of the equilibrium states can be calculated in a quite explicit form which depends only on the partition function and on the mean:

\begin{prop}
The free energy of an equilibrium state~\eqref{e:inv-meas-mckean}-\eqref{e:self-consist} is given by
\[
\mathcal{F}[p] = -\beta^{-1}\log Z + \frac{\theta}{2} m^2.
\]
In particular, when $m=0$ we have
\[
\mathcal{F}[p] = -\beta^{-1}\log Z.
\]
\end{prop}

\begin{proof}

The free energy of a function $p$ is given by
\begin{equation*}
\mathcal{F}\left[p\right] = \beta^{-1}\int_\R p(x) \log(p(x)) \ dx+\int_\R V\left(x\right)p(x) \ dx + \frac{\theta}{2}\int_\R \int_\R W(x-y) p(x) p(y) \ dx \ dy.
\end{equation*}
The stationary solution(s) to the Fokker-Planck equation~\eqref{e:mckean-vlasov} are given by Eqs.~\eqref{e:inv-meas-mckean}-\eqref{e:self-consist}. Plugging this into
the expression for the free energy, we obtain
\begin{eqnarray}\label{e:free}
\mathcal{F}\left[p\right] &=& \beta^{-1}\int_\R p(x) \left(-\beta\left(V(x)+\theta (W\star p)\right)-\log Z\right) \ dx+\int_\R V\left(x\right)p(x) \ dx + \\
\nonumber & & + \frac{\theta}{2}\int_\R \int_\R W(x-y) p(x) p(y) \ dx \ dy \\
&=&- \beta^{-1}\log Z - \frac{\theta}{2}\int_\R\int_\R W(x-y)p(x)p(y) \ dx \ dy,
\end{eqnarray}
where we have used that $\int_\R p(x) \ dx = 1$ and the definition of convolution.

Replacing $W(x-y) = \frac{(x-y)^2}{2}$ and using
Eq.~\eqref{e:inv-meas-mckean}, yields
\begin{equation}
\mathcal{F}\left[p\right] =-\beta^{-1}\log Z + \theta R(m)\left(m - \frac{1}{2}R(m)\right).
\end{equation}

We note that when $p$ is an equilibrium solution of the Fokker-Planck equation, then $R(m) = m$ and we obtain
\[
\mathcal{F}[p] = -\beta^{-1}\log Z + \frac{\theta}{2} m^2,
\]
and, when $m=0$ we recover
\[
\mathcal{F}[p] = -\beta^{-1}\log Z.
\]

\end{proof}

\subsubsection{Arclength continuation}
\label{Continuation} The second method we use is arclength continuation of
solutions, for which we will use the Moore-Penrose quasi arclength
continuation algorithm. Rigorous mathematical construction of the full
arclength continuation methodology can be found
in~\cite{Krauskopf,ContAllgower}. Some useful practical aspects of
implementing arclength continuation are also given in the \textsf{MATLAB}
manual~\cite{matcont}. The idea is to solve the discretized nonlinear
algebraic equation~\eqref{e:inv-meas-mckean}-\eqref{e:self-consist} for a
given initial value of the control parameter, $\beta_0$, and a given initial
guess, $m_0$, relaxing the dependence on $\beta_0$ and adding a condition of
curve continuity in the phase space of solutions to the discretized problem.
The method then provides us with a way of following each branch by computing
tangent vectors.

We use arclength continuation to construct the bifurcation diagram of steady
solutions of~\eqref{e:mckean-vlasov} using the self-consistency
equation~\eqref{e:self-consist} and the system of equations for the moments
described in the next section. Arclength continuation is also used to solve
the equation for the critical temperature $\beta_C$ given by the solution
of~\eqref{eq:beta_C} as a function of $\theta$.

\subsubsection{Time dependent simulations}
To simulate the corresponding particle system, we perform MC simulations of
$N=1000$ particles evolving according to the system of
SDEs~\eqref{e:sde-general}. We use the Euler-Maruyama numerical scheme, with
time step $dt = 0.01$, and solve for long times to ensure that the solution
has converged to its invariant measure.

We also solve numerically the Fokker-Planck equation~\eqref{e:mckean-vlasov},
subject to the boundary conditions of zero particle flux through the
boundaries of our numerical interval. We approximate the derivative with a
pseudo-spectral Chebyshev collocation method, and the integral term with a
Clenshaw-Curtis quadrature \cite{Peter_2012}. For marching in time, we adopt
the \textsf{ode15s} function of \textsf{MATLAB}, which is based on an
implicit scheme combining backward differentiation and adaptive time
stepping.
%
%
\section{Results of numerical simulations}
\label{sec:numeric} We now resent numerical results for bifurcation diagrams
and certain time-dependent simulations using the methodologies described
above. The results are obtained for the three types of polynomials listed in
Section~\ref{subsec:models}.
%
%
\subsection{Polynomial potentials}\label{s:num_poly}
%
Here we consider confining potentials of the form
\begin{equation}\label{eq:polyn_potential}
V(x) = \sum_{\ell=1}^{M} a_\ell x^{2\ell},
\end{equation}
where $M = 2,3,\dots$. This introduces additional wells in the confining
potential, corresponding to different local minima-maxima in the potential,
which in turn translates into various pitchfork and/or saddle-node
bifurcations from the mean-zero solution, with the corresponding changes in
stability, as will be seen below.

As mentioned earlier, if $V(x)$ is a polynomial, we can obtain a system of
ODEs verified by the moments $f(x) = x^k$, in a similar manner to what was presented 
in~\cite{Dawson1983} for the bistable potential $V(x) = \frac{x^4}{4}-\frac{x^2}{2}$ 
and which easily extends to arbitrary polynomial potentials.  
To this end, we consider the system of SDEs~\eqref{e:sde-general} with $W(x) =
\frac{x^2}{2}$, and by defining $m_k(t) = \frac{1}{N}\sum_{\ell = 1}^N
\left(X_t^\ell\right)^k$, we rewrite it as
\begin{equation}\label{eq:SDE_System_for_moments}
dX_t^i = -V'(X_t^i) \ dt + \theta (m_1(t) - X_t^i) \ dt + \sqrt{2\beta^{-1}} dB_t^i,
\end{equation}
 $i = 1,\dots, N$. Using It\^{o}'s Lemma,
we can obtain a system of SDEs for the moments $f(x) = x^k$:
\begin{equation}\label{Ito_Lemma}
 dx^k(t) =  k\left[-\theta x^k(t) + \left(\theta m_1(t) - V'(x(t))\right)x^{k-1}(t)
+ \beta^{-1}(k-1)x^{k-2}(t) \right] dt + \sqrt{2\beta^{-1}} k x^{k-1} dw(t),
\end{equation}
where $w(t)$ is white noise. Replacing $V$ by its expression, noticing that $m_0(t)=1$,
and taking expectations, we obtain a system of ODEs for $m_k(t)$, $k=1,2,\dots,\infty$.
Unfortunately, due to the structure of the potentials and the nonlinearity
involved, this cannot be expanded for other types of potentials.

In the so-called ferromagnetic case, $V(x) = V_4(x) = \frac{x^4}{4}-\frac{x^2}{2}$, we obtain the following system of ODEs:
\begin{equation}\label{eq:Moments_4thDegree}
\dot{m}_k(t) = k\left((1-\theta) m_k(t) + \theta m_1(t)m_{k-1}(t) + \beta^{-1}(k-1)m_{k-2}(t) - m_{k+2}(t)\right)
\end{equation}
Other examples include higher degree polynomials:
\begin{align}
\label{eq:6thDegree}V_6(x) &= h\left(x^6 - 5x^4 +4x^2\right) = hx^2(x^2-1)(x^2-4),\\
\label{eq:8thDegree}  V_8(x) &= h\left(x^8 - 14x^6+49x^4 - 36x^2\right)  = hx^2(x^2-1)(x^2-4)(x^2-9),
\end{align}
where we have added a pre-factor $h$ in the higher degree polynomials. This is to make the barrier at $x=0$ (and/or others) more relevant. For the $6^{th}$ degree case, we obtain the following system of ODEs for the moments $m_k$:
\begin{equation}\label{eq:Moments_6thDegree}
 \dot{m}_k = k\left(-\left(8h+\theta\right) m_k + \theta m_1m_{k-1}
  + \beta^{-1}(k-1)m_{k-2} + 20hm_{k+2} - 6hm_{k+4}\right),
\end{equation}
and in the $8^{th}$ degree case
\begin{equation}\label{eq:Moments_8thDegree}
\dot{m}_k = k\left(\left(72h-\theta\right) m_k + \theta m_1m_{k-1}+ \beta^{-1}(k-1)m_{k-2}   - 196hm_{k+2} + 84hm_{k+4}-8hm_{k+6}\right).
\end{equation}
We truncate the system at $k=21$ and solve for the first moment, performing
arclength continuation. 

As an illustration of our methods, we plot in Fig.~\ref{FigFE1} a bistable
potential $V(x) = \frac{x^4}{4}-\frac{x^2}{2}$ and the corresponding
bifurcation diagram of $m$ as a function of $\beta$. We used both arclength
continuation for the self-consistency equation and the method of moments, as
well as the free-energy method, obtaining similar results in all cases.
\begin{figure}[h!]
\centering
    \includegraphics[scale=1.25]{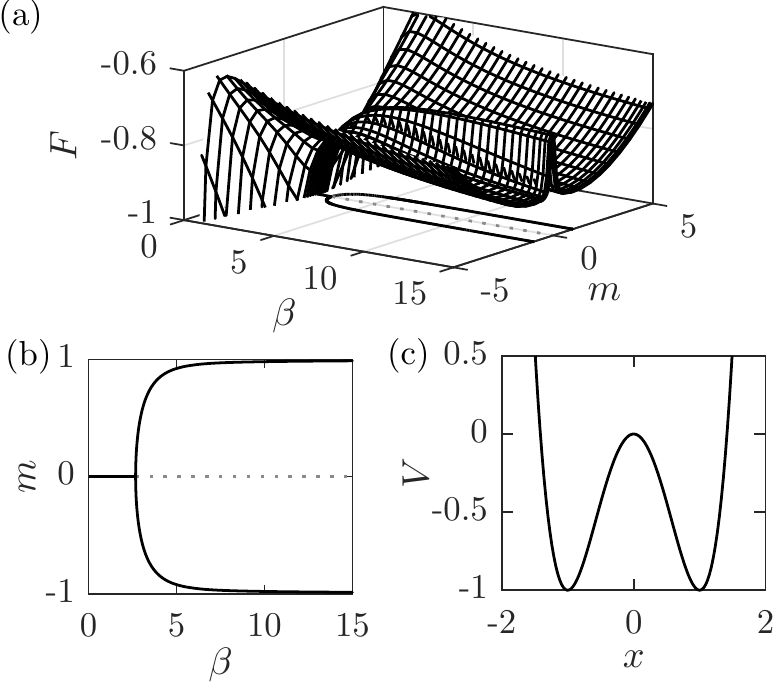}
    \caption{Free energy surface (a) and the corresponding bifurcations of the steady states of the Fokker-Planck equation (b) in a simple bi-stable potential $V(x) = \frac{x^4}{4}-\frac{x^2}{2}$ (c). The control parameter is $\beta$, and the solution norm is given by the first moment $m$, with $\theta = 0.5$. In (a) and (b), solid and dotted lines correspond to stable and unstable steady states of the Fokker-Planck equation, respectively.
    \label{FigFE1}}
\end{figure}

In Fig.~\ref{FigFE1}, we illustrate our free-energy method to obtain the
bifurcation diagram. We fix $\theta = 0.5$ and compute the free-energy
surface for functions $p(x;m,\theta,\beta)$ given
by~\eqref{e:inv-meas-mckean} (without assuming that $m$ verifies
\eqref{e:self-consist}) and proceed to compute its extrema, which are
contoured below.

We present one more example of a polynomial confining potential, where $V(x)
= V_8(x)$ from Eq.~\eqref{eq:8thDegree}. We fix $h = 0.001$ and compute the
bifurcation diagram of $m$ as a function of $\beta$ for $\theta = 1.5$
(\ref{FigFE8}(a)) and $\theta = 2.5$ (\ref{FigFE8}(b)). We observe that the
topology of the bifurcation diagram is different for the two values of
$\theta$: for small $\theta$ the effects of the interaction do not affect the
convexity of the free energy and we find three pitchfork bifurcations from
the mean-zero solution with the corresponding (expected) changes in
(linear/local) stability of all the solutions. However, for large enough
$\theta$, the effects of the interaction change the convexity of the
free-energy functional, and we find only one pitchfork bifurcation,
accompanied by two saddle-node bifurcations. We observe, however, that the
number of solutions for large $\beta$ is still $7$, which is the number of
equilibrium points of the confining potential $V_8$. In fact this is found
for all the confining potentials studied in this work.
\begin{figure}[h!]
\centering
    \includegraphics[width=0.85\linewidth]{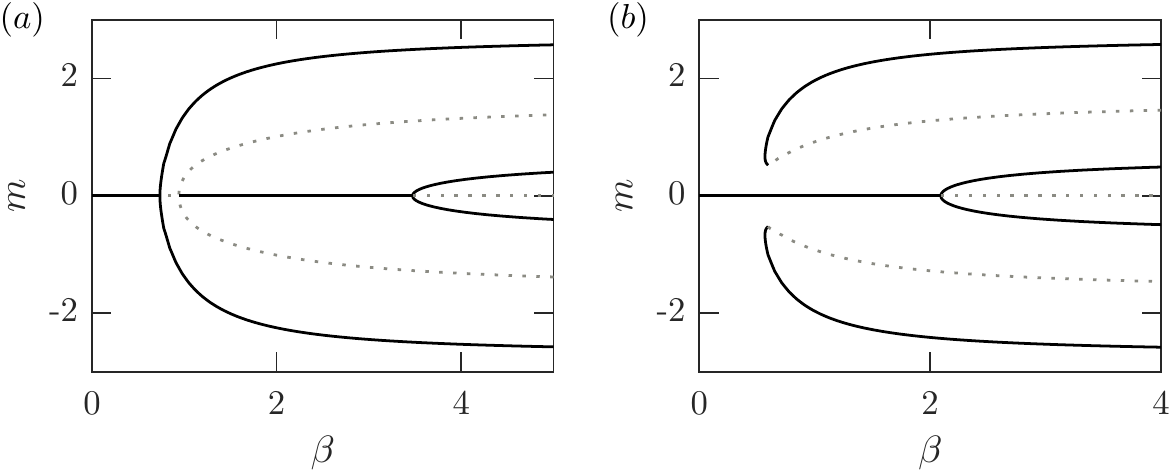}%
    \caption{Phase diagram in the $\beta$--$m$ space, for the potential $V_8(x)$, at $\theta=1.5$ (a) and
$\theta=2.5$ (b). Solid and dotted branches are, respectively, stable and unstable steady states of the Fokker-Planck
equation. \label{FigFE8}}
\end{figure}

The effect of $\theta$ in the topology of the bifurcation diagram can be
further analyzed by studying the critical temperature $\beta_C$ at which a
pitchfork bifurcation from the mean-zero solution occurs,  as a function of
$\theta$ . This is given by Eq.~\eqref{eq:beta_C}. We solve this equation for
$\beta_C$ as a function of $\theta$ again performing arclength continuation,
and plot the results in Fig.~\ref{fig:critical_temp}.

\begin{figure}[h!]
\centering
    \includegraphics[width=\linewidth]{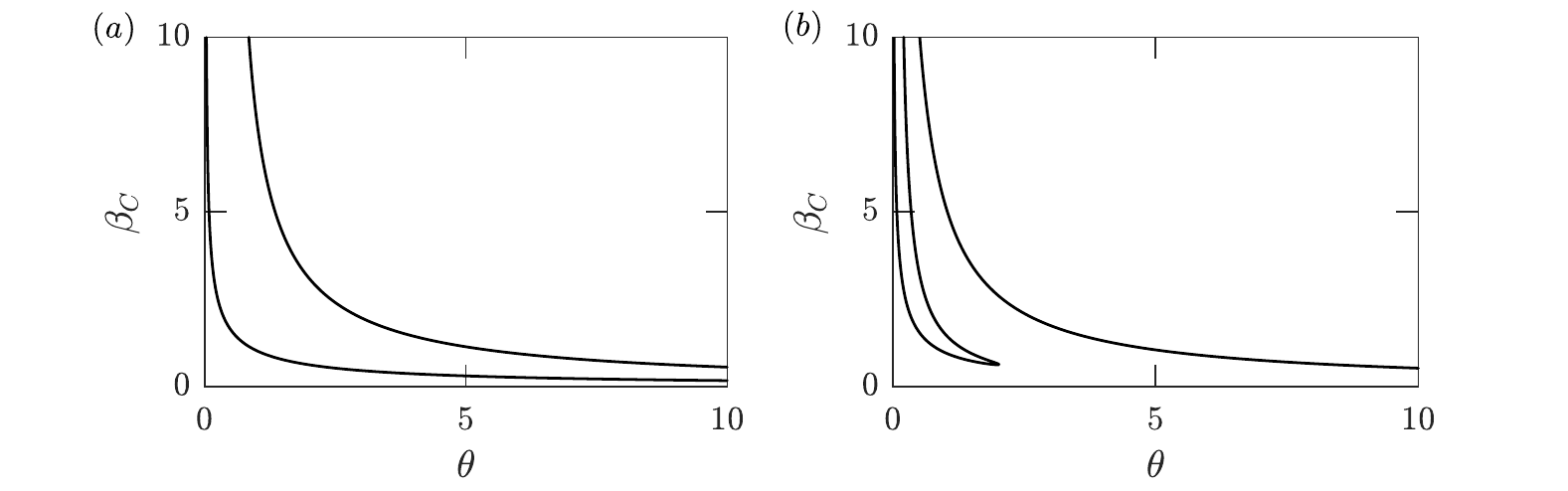}%
    \caption{Critical temperature $\beta_C$ as a function of $\theta$ for the potentials $V_6(x)$ with $h = 0.1$
(left panel) and $V_8(x)$ with $h = 0.001$ (right panel) given by
Eqs~\eqref{eq:6thDegree} and~\eqref{eq:8thDegree} respectively.
    \label{fig:critical_temp}}
\end{figure}

We also see the existence of two branches for the $6^{th}$ degree polynomial
potential, which corresponds, as expected, to the existence of two pitchfork
bifurcations in the bifurcation diagram of the first moment, $m$, as a
function of the inverse temperature $\beta$. Interestingly, for the $8^{th}$
degree polynomial, there are three branches of solutions for sufficiently
small $\theta$ but these branches merge for $\theta \approx 2$. This
indicates that the convexity of the free-energy functional changes as a
function of $\theta$, which also means that for polynomial potentials it is
important to keep track of the bifurcation structure as a function of both
$\beta$ and $\theta$. This change of convexity leads to the behavior observed
in Fig.~\ref{FigFE8}: for small enough $\theta$ (Fig.~\ref{FigFE8}(a)) there
exist three pitchfork bifurcations, with the corresponding change of
stability in the mean-zero solution, while for large values of $\theta$
(Fig.~\ref{FigFE8}(b)) there is only one pitchfork bifurcation, with the
$m=0$ solution remaining the global minimum of the free-energy for larger
values of $\beta$. The other stationary solutions still exist, but they
appear as discontinuous bifurcations (corresponding to first-order phase
transitions).

Finally, we study the effect of breaking the symmetry of polynomial
potentials by adding a tilt to a bistable potential. Specifically, we
consider potentials of the form
\begin{equation}\label{eq:tilted_bistable}
V(x) = \frac{1}{a_0}\left(\frac{x^4}{4}-\frac{x^2}{2}\right) + \kappa x.
\end{equation}
Fig.~\ref{FigDyn0} depicts the bifurcation diagram of $m$ as a function of
$\beta$ for this potential, with $\theta = 2.5$, $a_0 = 0.249998581434761$,
and $\kappa = 0, \, 0.01, \, 0.1, \, 1$. Evidently a break in the symmetry of
the bifurcation diagram appears, which becomes increasingly clear as $\kappa$
increases.
\begin{figure}[h!]
\centering
    \includegraphics[scale=1]{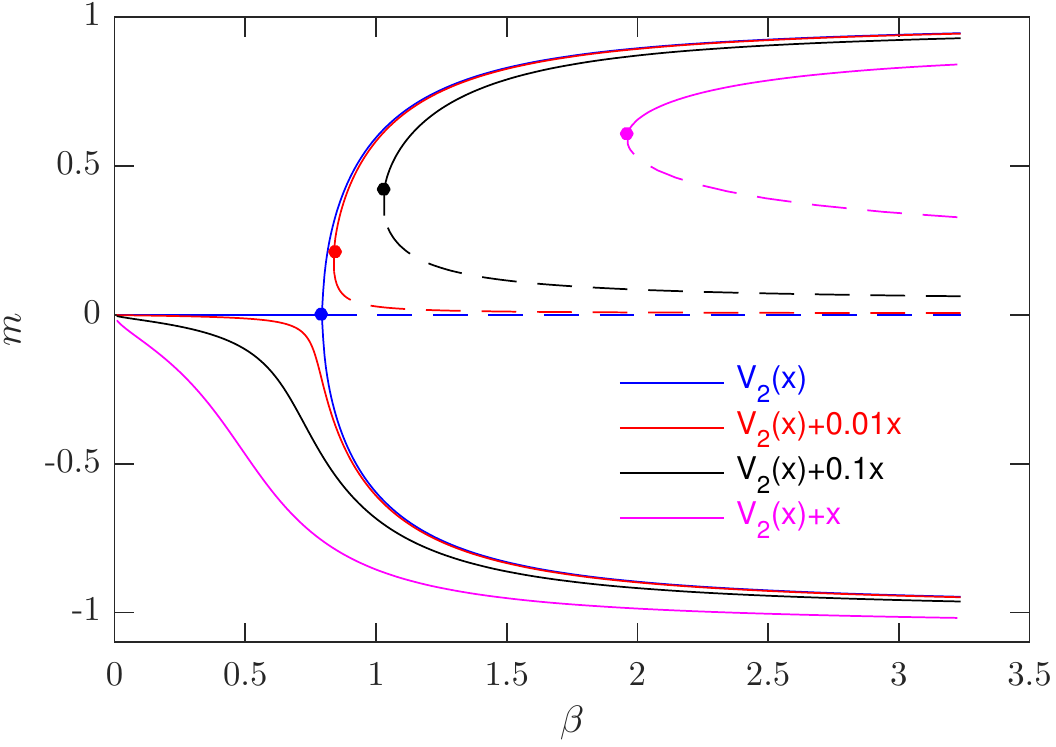}%
    \caption{Bifurcation diagrams of $m$ as a function of $\beta$ for tilted bistable potentials given by Eq.~\eqref{eq:tilted_bistable} with
$a_0 = 0.249998581434761$,  and $\kappa = 0, \, 0.01, \, 0.1, \, 1$ (see legend). Here, we used $\theta = 2.5$ The symmetric pitchfork bifurcation is
broken at any $\kappa>0$. We note that the locus of the critical points forms a distinct critical line.
    \label{FigDyn0}}%
\end{figure}

The existence of stationary solutions raises the question of relevance of
these solutions which is related to the way they attract initial conditions.
An answer to this question can be given by means of time-dependent
computations. Figs.~\ref{fig:tilted1} and~\ref{fig:tilted2} depict the time
evolution of the first moment as a function of time (top panel) and the
histogram and corresponding distribution (solution of the time dependent
Fokker--Planck equation) in the bottom panel, at two selected times marked in
dashed lines in the top panel. Both figures correspond to the tilted bistable
potential given by Eq.~\eqref{eq:tilted_bistable} with $\kappa=0.1$, $\theta
= 2.5$ and $\beta = 1.5$. Fig.~\ref{fig:tilted1} shows the evolution
starting from a $N(0.1,\beta^{-1})$, while in Fig.~\ref{fig:tilted2} the time
evolution is started from a $N(-0.1,\beta^{-1})$ distribution.
\begin{figure}[h!]
\centering
    \includegraphics[width=0.75\linewidth]{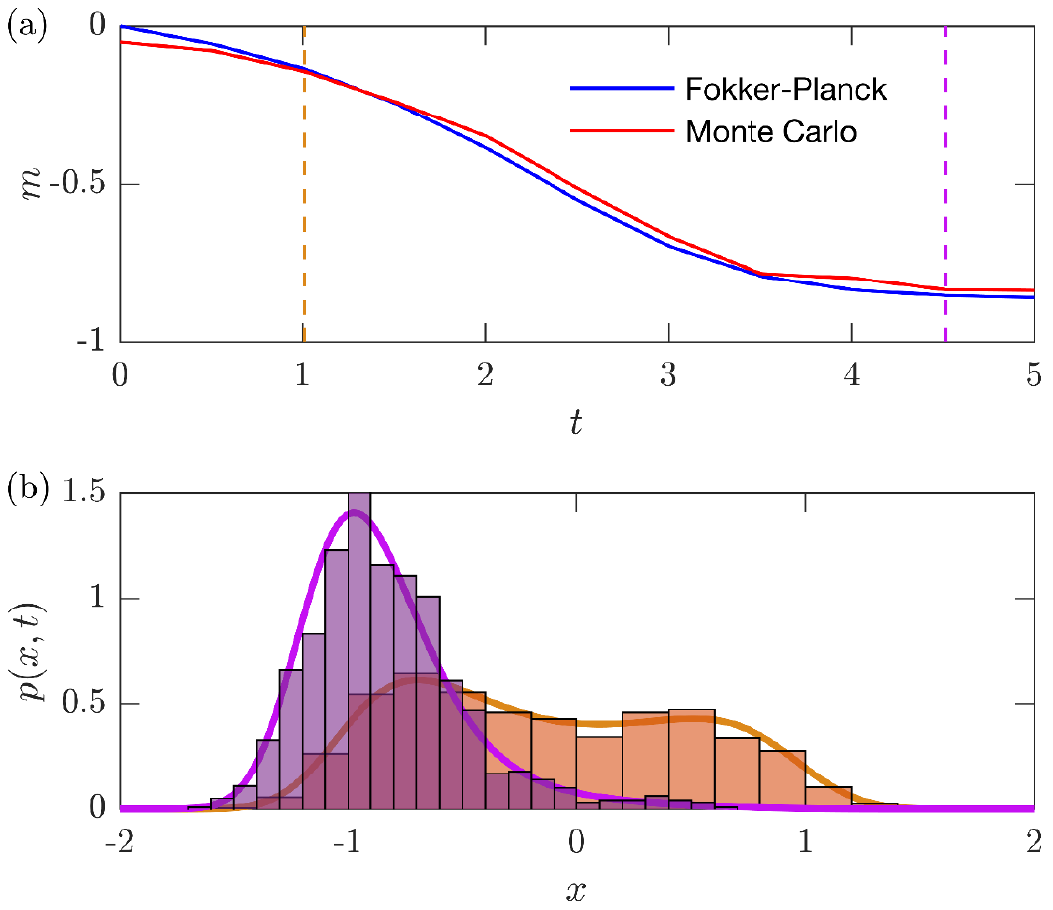}
    \caption{Evolution of the density and first moment in a tilted bistable potential given in Eq.~\eqref{eq:tilted_bistable},
    with $\kappa=0.1$ and an initial condition distributed according to a
    $N(0,\beta^{-1})$ distribution.
    (a) Mean position against time for the evolution of the Fokker--Planck equation (blue line) and for the
interacting particles system (red line). (b)
    Fokker-Planck distributions and corresponding MC histograms for
    selected times, designated in (a) by vertical dashed lines of
    respective colors (see also the Supplemental Material movies \texttt{MovM1.avi}, showing simultaneously the first moment on the bifurcation diagram and the distribution, and \texttt{MovFig8.avi}, showing the good agreement between Fokker-Planck and MC simulations).}
    \label{fig:tilted1}
\end{figure}
\begin{figure}[h!]
\centering
    \includegraphics[width=0.75\linewidth]{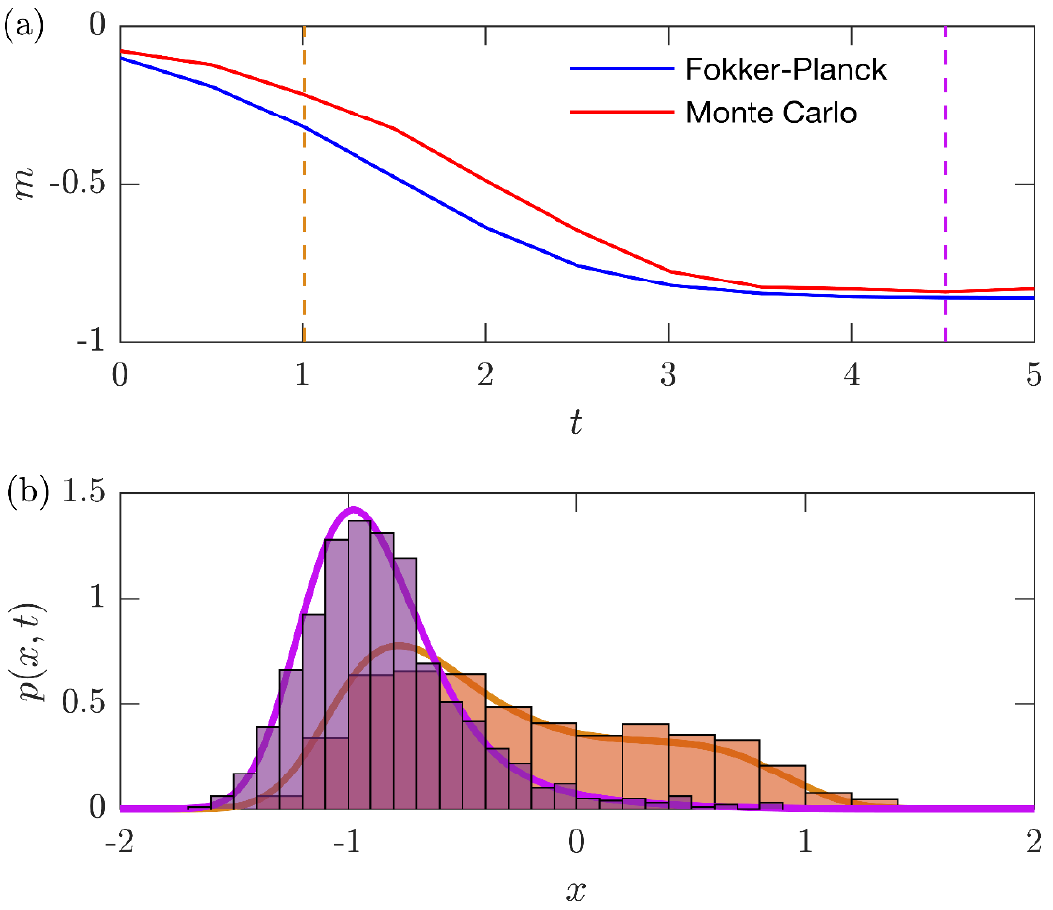}
    \caption{Evolution of the density and first moment in a tilted bistable potential given in Eq.~\eqref{eq:tilted_bistable}, with $\kappa=0.1$ and an
    initial condition distributed according to a $N(-0.1,\beta^{-1})$ distribution.
    (a) Mean position against time for the evolution of the Fokker--Planck equation (blue line) and for the
interacting particles system (red line). (b) Fokker-Planck distributions and
    corresponding Monte-Carlo histograms for selected times, designated
    in (a) by vertical dashed lines of respective colors (see also movie MovFig9.avi in Supplemental Material showing
    good agreement between Fokker-Planck and MC simulations).
    \label{fig:tilted2}}
\end{figure}

An overall good agreement is observed between the solution of the time
dependent Fokker--Planck equation and the corresponding MC simulations. It is
worth noting that, without the tilt, the dynamics reproduced in
Fig.~\ref{fig:tilted1} would have evolved to the upper branch of the
bifurcation diagram represented in Fig.~\ref{FigDyn0}, but instead we
observe the breaking of symmetry caused by the tilt: the particles would have
to pass through an unstable equilibrium point (represented by a black dashed
line in Fig.~\ref{FigDyn0}) in order to reach the upper branch.

%

\subsection{Rational potentials }\label{s:num_zofia}
%
Here, we consider potentials of the form
\begin{equation}\label{eq:zofia}
V(q) = \frac{1}{\sum_{\ell=-N}^N \delta_{\ell} |q - c_{\ell} q_{\ell} |^{-2}}.
\end{equation}
with both deterministic and random distributions of
$\left\{\delta_\ell,c_\ell\right\}$. We show two examples in particular. The first one is a
potential with $6$ minima, which are symmetrically located and have the same
depths (as well as heights of the corresponding local maxima),
\begin{equation}\label{6_zeros}
V(x) = h\left((x-1)^{-2} + (x+1)^{-2} + (x-2)^{-2}+(x+2)^{-2} + (x-3)^{-2}+(x+3)^{-2}\right)^{-1}.
\end{equation}
The free energy surface and bifurcation diagram for this case are presented
in Fig.~\ref{FigFE6}.

\begin{figure}[h!]
\centering
    \includegraphics[width=0.75\linewidth]{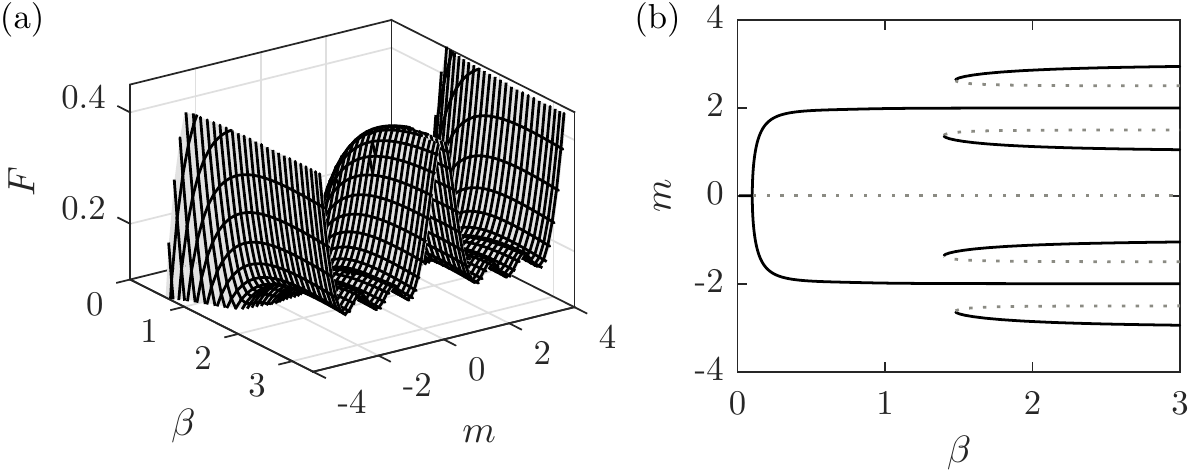}
    \caption{Free energy surface (a) and the corresponding bifurcation diagram in the $\beta$--$m$ parameter space (b) for the potential
    in Equation \eqref{6_zeros} with $\theta=5$ and $h=1$. Stable and unstable
    branches are plotted with solid and dotted lines respectively.
    \label{FigFE6}}
\end{figure}

We now consider the potential from Eq.~\eqref{eq:zofia} with $N=20$
minima positioned at $x=-10,-9,\dots10$, and $c_\ell=1$, $\ell=1,\dots,20$.
Fig.~\ref{FigBif0} depicts the realization of a random potential, where the
energy barriers separating the local minima of the potential are uniformly
distributed random variables, i.e. $\delta_\ell\sim U([0,1])$.

\begin{figure}[h!]
\centering
    \includegraphics[width=0.75\linewidth]{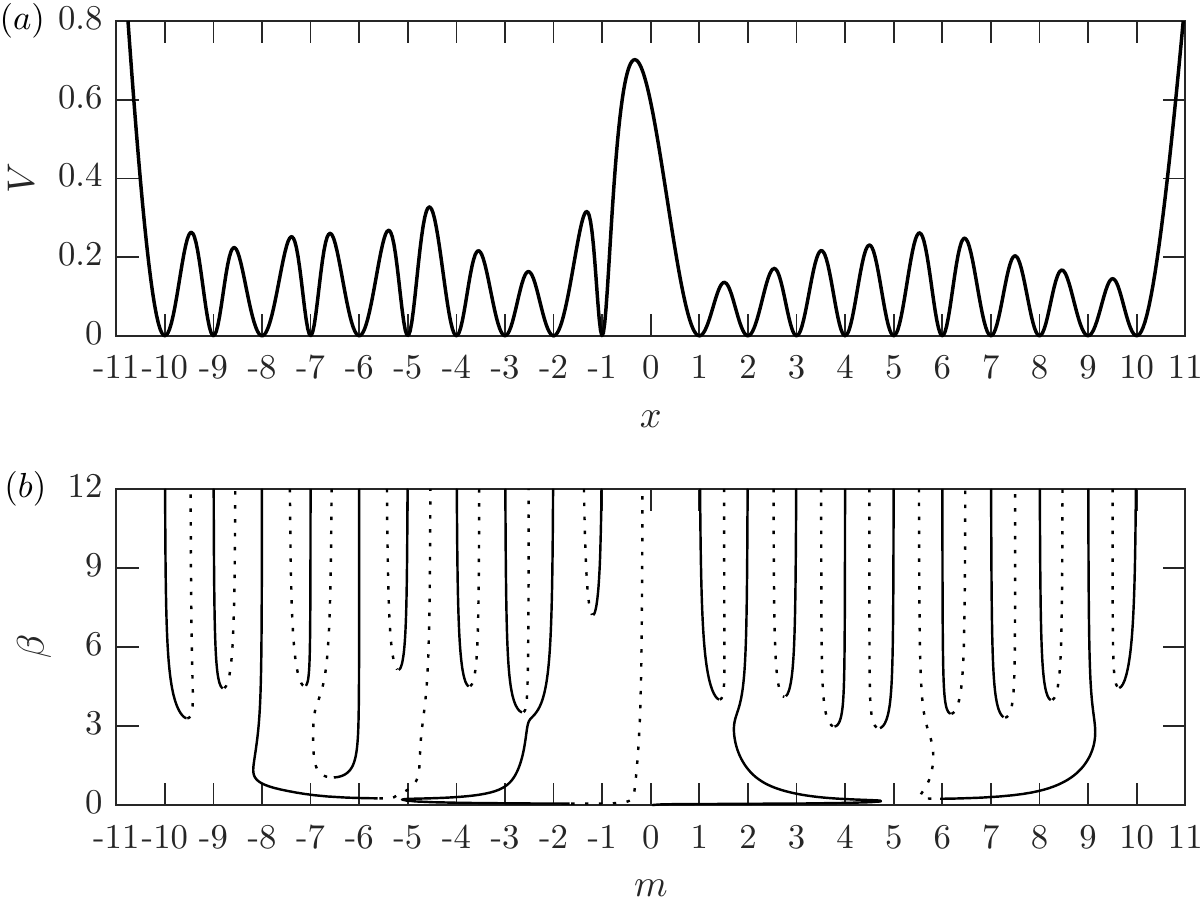}%
    \caption{(a) A potential from Equation \eqref{eq:zofia}, with the minima positioned at integers between -10 and 10, separated by local maxima of arbitrary heights.
    (b) Corresponding bifurcation diagram of the steady states of the Fokker-Planck Equation \eqref{e:mckean-vlasov}, obtained from the computed free-energy surface.
    Stable and unstable states are designated by solid and dashed curves, respectively.
    \label{FigBif0}}%
\end{figure}

Interestingly, the random depths of each local minima -- these correspond to
higher or lower energy barriers -- affect the stability of each well with
respect to each other. Further insight into the effect of the random depths
on the dynamics of the system can be obtained via the time-dependent
evolution of both the particle system and the mean-field Fokker-Planck
equation. The corresponding results are plotted in Figs.~\ref{FigDyn1}
and~\ref{FigDyn2}.
\begin{figure}[h!]
\centering
    \includegraphics[width=0.75\linewidth]{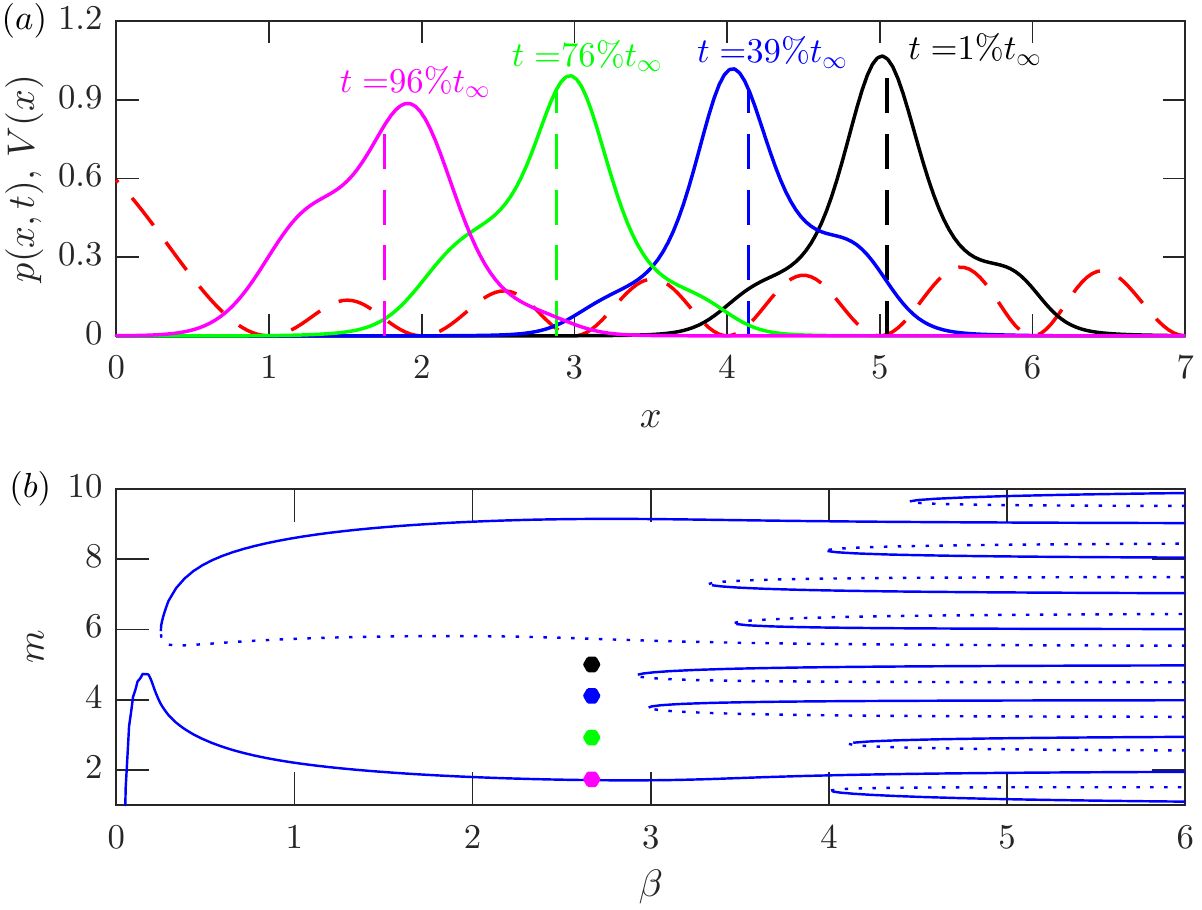}%
    \caption{Numerical solution of the Fokker-Planck equation \eqref{e:mckean-vlasov}, for the potential given in Fig.~\ref{FigBif0}
    (see also \texttt{MovFig12.avi} in the supplemental material).
    (a) $p(x,t)$ (solid curves) and $V(x)$ (dashed curve). The first moments for each $p(x,t)$ are designated by dashed verticals,
    and the respective values of $t$ are provided. (b) States from (a) on the bifurcation diagram.
    We note that the basins of attraction of the stable states are effectively demarcated by the metastable branches of the bifurcation diagram.
    Here $t_{\infty}$ is the time upon which the relative norm of the difference between the steady state and the time-dependent solution is less than 5\%.
    \label{FigDyn1}}%
\end{figure}

Fig.~\ref{FigDyn1} shows the solution of the Fokker--Planck equation as a
function of time for four different times (top panel), and the
corresponding location of the mean of the solution at each of these times
on the bifurcation diagram (lower panel).
We have introduced an empirical ``convergence'' time $t_{\infty}$, defined as the time at which the relative norm between the time-dependent solution and the steady state is less than 5\%:
\begin{equation}
{\left\|p_{\infty}(x)-p(x,t_{\infty})\right\|}\leq 0.05{\left\|p_{\infty}(x)\right\|},
\end{equation}
where $\|\cdot\|$ stands for the Euclidean norm over $\mathbb{R}$. It should
be noted that as we approach the turning points of the bifurcation diagram in
Fig.~\ref{FigDyn1}(b), by e.g. fixing $\beta$ and increasing $m$, the system
slows down and eventually ``freezes" and gets pinned to the branch of
metastable solutions, which terminates at the turning point. Hence, in the
neighborhood of turning points, we have a ``glass"-like behavior
(e.g.~\cite{Berthier-Biroli_2011}), and the potential in Fig.~\ref{FigBif0}
can be viewed as a ``glassy potential". Movie \texttt{MovM1.avi} in the Supplemental Material 
shows such a pinning transition for a tilted bistable potential, which clearly occurs at the boundary 
of the basin of attraction of the upper U-branch. 
For model prototype systems, such as the Swift-Hohenberg
equation, the time between two consecutive transitions can be estimated via
weakly-nonlinear analysis in the vicinity of the turning
points~\cite{Burk-Knobloch_2006}. But our equations are too involved to be
amenable to analytical treatment of this type.

We can explore the behavior in Fig.~\ref{FigDyn1} further by plotting the
first moment of the solution as a function of time -- we do so in
Fig.~\ref{FigDyn2}. The top panel displays the first moment as a function
of time for the solution of the Fokker--Planck equation (full blue line)
compared with two independent runs of the particle system. while the bottom
panel compares the Fokker--Planck solution with the histograms from the
particle simulations. Each particle run had $N=1000$ particles and the
simulations shown use $\theta = 1.5$ and $\beta\approx 2.66$.
\begin{figure}[h!]
\centering
    \includegraphics[width=0.75\linewidth]{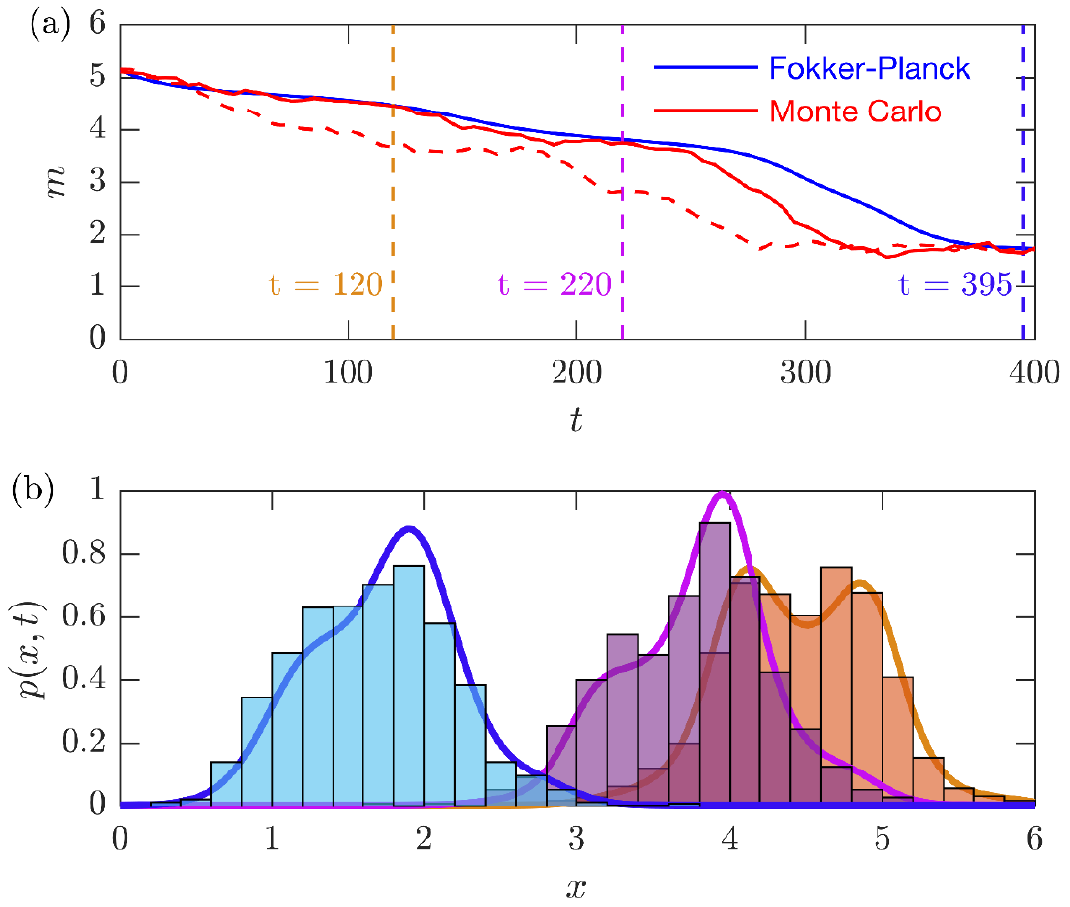}%
    \caption{Same as in Fig.~\ref{FigDyn1}, against the corresponding Monte--Carlo simulation data. (a) Mean position against time for the evolution of the Fokker--Planck equation (full blue line) and the interacting particles system for two realisations of the noise (full and dashed red lines).
    (b) Fokker-Planck distributions and corresponding Monte Carlo histograms for selected times, designated in (a) by vertical dashed lines of respective colors.
    We can see an overall good agreement between the simulations and computations.
    \label{FigDyn2}}%
\end{figure}

%
\subsection{Piecewise linear potentials with quadratic growth at infinity}\label{sec:P-L}
%

Here we will replace the confining potential with a piecewise linear
approximation with quadratic growth at infinity. Our motivation for this is
that we can now compute the partition function $Z(m,\theta,\beta)$, the mean
$R(m)$ and the variance analytically (in terms of the error function ${\rm
erf}(x) = \frac{2}{\sqrt\pi}\int_0^x e^{-x^2} \ dx$). We consider a potential
$V$ with $2M$ local minima, $X = (x_1,\dots,x_{2M})$, with depth $h_i$, and
the local maxima, located at $x_{i+\frac{1}{2}} = \frac{x_i+x_{i+1}}{2}$
(note that $x_{-1/2} = 0$), have height $H_i$. In this case, $V$ is given by
\begin{equation}\label{eq:P-L-General} V(x) =
\left\{\begin{array}{lr}
\frac{x^2 -x_1^2}{2} + h_1,& \tiny{x\leq x_1}, \\
\frac{(H_i-h_i)x}{x_{i+\frac{1}{2}}-x_i} +\frac{h_ix_{i+\frac{1}{2}}-H_ix_i}{x_{i+\frac{1}{2}}-x_i},& x_i<x<x_{i+\frac{1}{2}}, \\
\frac{(h_i-H_i)x}{x_{i+1}-x_{i+\frac{1}{2}}} + \frac{H_ix_{i+1}-h_ix_{i+\frac{1}{2}}}{x_{i+1}-x_{i+\frac{1}{2}}},  & x_{i+\frac{1}{2}}<x<x_{i+1},\\
\frac{x^2 -x_{2M}^2}{2} + h_{2M}&  x\geq x_{2M}.
\end{array}\right.
\end{equation}
$i=1,\dots,2M$.

As mentioned before, we can compute the quantities $Z(m,\theta,\beta), \,
R(m)$ and the second moment analytically, and will present two illustrative
cases: a  symmetric potential with $2M=6$ wells and a non-symmetric potential
with random heights and depths. Details of the analytical calculations are
given in Appendix~\ref{ap:1}.

\subsubsection{Symmetric potentials with six wells at same heights}\label{sec:P-L-6}

Below we present the results for $V(x)$ given by Eq.~\eqref{eq:P-L-General}
with $M=3$. We use $X = (-3,-2,-1,1,2,3)$, $x_{i+\frac{1}{2}} =
\frac{x_i+x_{i+1}}{2}$, $h_i = 0$, and $H_i=1$. We compute $Z(m,
\theta,\beta)$ and $R(m)$ using Eqs~\eqref{eq:ap_Z} and~\eqref{eq:ap_M}
respectively, and solve for $R(m)=m$ using arclength continuation. We plot
our results in Fig.~\ref{FigFE100}. Here, we choose $\theta = 5$. Panel (a)
shows the solution of the self-consistency equation $R(m)=m$ (or, rather
$R(m)-m=0$) for $\beta = 1$ and $10$. Panel (b) shows the bifurcation diagram
of $m$ as a function of $\beta$ and (c) shows the critical temperature
$\beta_C$ as a function of $\theta$, which was obtained using
Eq.~\eqref{eq:ap_V} for $m=0$.

\begin{figure}[h!]
    \includegraphics[width=1.125\linewidth,trim={33 0 0 0},clip]{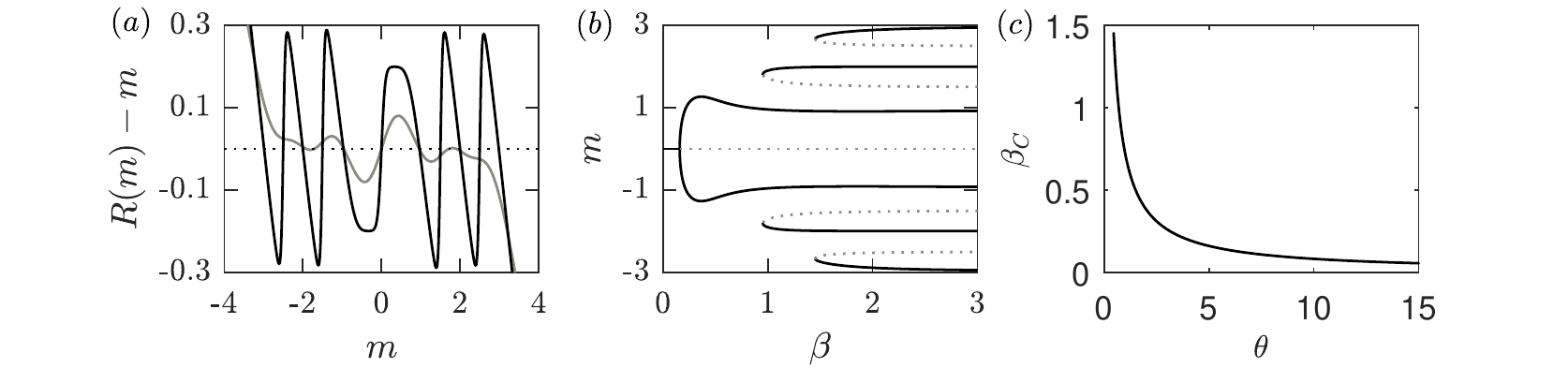}%
    \caption{Phase diagram for the case of the piecewise linear potential \eqref{eq:P-L-General} with 6 wells for $\theta = 5$. Panel (a) shows the solution of $R(m)-m=0$
for $\beta=1$ (black) and $\beta=10$ (grey). Horizontal dotted line is drawn at $R(m)-m=0$. Panel (b) shows the bifurcation diagram in the $\beta$--$m$ space.
Panel (c) shows the critical temperature $\beta_C$ at which a pitchfork bifurcation occurs from the mean-zero solution as a function of $\theta$.
    \label{FigFE100}}%
\end{figure}

The results are what we would expect: there are eleven equilibrium points
which correspond to each local minimum and maximum. We point out the
similarity between the bifurcation diagram in panel (b) to the one presented
in Fig.~\ref{FigFE6}, which shows that a piecewise linear potential is a good
first approximation.

\subsubsection{Potentials with four wells at different (randomly distributed) heights}\label{sec:P-L-4-R}
Our final test is the case where $V(x)$ is given by
Eq.~\eqref{eq:P-L-General} with $M=2$, but with the minima and maxima heights
and depths randomly distributed. We use,  $X = (-2,-1,1,2)$,
$x_{i+\frac{1}{2}} = \frac{x_i+x_{i+1}}{2}$, and generated $h_i$ and $H_j$,
$i=1,\dots,4$, $j=1,\dots 3$ randomly, following a uniform distribution.

As before, we compute the relevant functions of $m$, $\theta$ and $\beta$
using Eqs~\eqref{eq:ap_Z}-\eqref{eq:ap_V} and depict in Fig.~\ref{FigFE10}
the solution of $R(m)-m=0$ for $\beta = 1, \, 10$ and $\theta = 5$ (panel
(a)), the bifurcation diagram of $m$ as a function of $\beta$ for $\theta =
5$ and the critical temperature $\beta_C$ as a function of $\theta$.
\begin{figure}[h!]
\centering
 \includegraphics[width=1.125\linewidth,trim={33 0 0 0},clip]{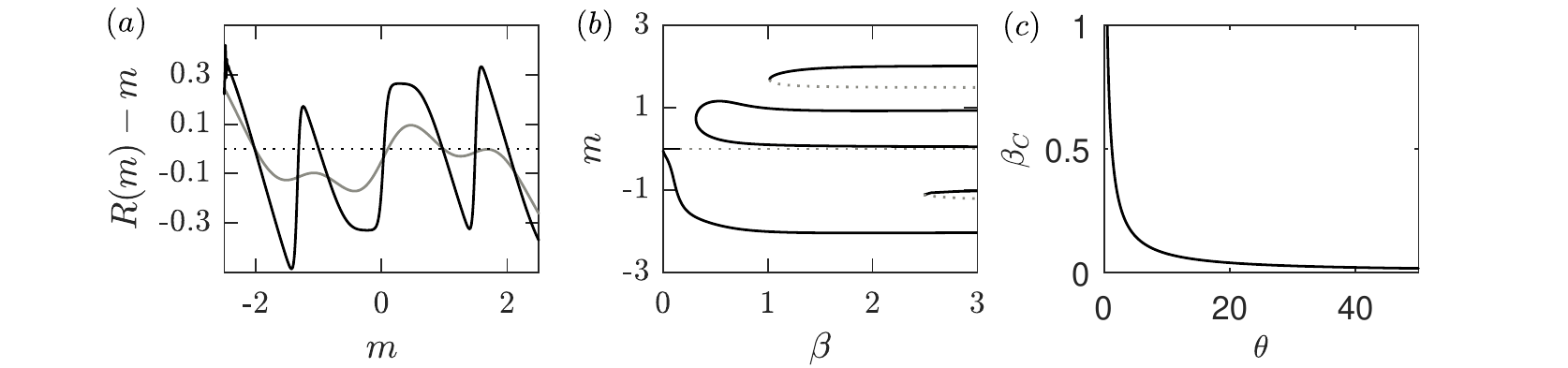}%
    \caption{Phase diagram for the case of the piecewise linear potential \eqref{eq:P-L-General} with 4 wells at random heights and depths for $\theta = 5$.
Panel (a) shows the solution of $R(m)-m=0$ for $\beta=1$ (black) and $\beta=10$ (grey). Horizontal dotted line is drawn at $R(m)-m=0$.
Panel (b) shows the bifurcation diagram in the $\beta$--$m$ space. Panel (c) shows the critical temperature $\beta_C$ at which a pitchfork bifurcation occurs
from the mean-zero solution as a function of $\theta$.
    \label{FigFE10}}%
\end{figure}

A behavior similar to that of the random potential in the previous section
is observed but with a smaller number of equilibrium points.

%
\section{Conclusions}\label{s:conclusions}
We presented a detailed and systematic investigation of the dynamics of a
system of interacting particles in one dimension, moving in a confining
multi-well potential and interacting through a Curie-Weiss/quadratic
potential. Passing formally to the mean-field limit yields the McKean SDE.
The Fokker-Planck equation corresponding to this SDE is the McKean-Vlasov
equation, which is nonlocal and nonlinear, and is the basic equation for our study. It is
a gradient-flow for a certain free-energy functional, establishing also a
connection with thermodynamics.

A wide spectrum of prototypical model potentials was considered: polynomial
(including tilted bistable ones), rational (both deterministic and random)
and piecewise linear potentials with quadratic growth at infinity that allow
for analytical estimates of the partition function, mean and variance. For
all these model potentials, we scrutinized steady states of the McKean-Vlason
equation, constructed bifurcation diagrams and studied their behavior in the
parameter space, and determined the stability of the solution branches. We
also determined the critical points and characterized the structure and
nature of phase transitions. We showed, by means of extensive computations --
including free-energy minimization, arclength continuation, simulations of
the full McKean-Vlasov equation, and MC simulations of the corresponding
particle system -- and of analytical calculations for explicitly solvable
models, that the number of steady states, their stability and structure of
bifurcation diagrams depends crucially on the form of the multi-well
potential and its characteristics, mainly the number and depth of the local
minima of the confining potential. Increasing the complexity of the potential
increases also the complexity of the steady-state bifurcation structure and
dynamics. Even the local minima of the potential may significantly affect the
relaxation dynamics of the system, via the basin of attraction of the
metastable states. Thus, each local minima of the potential gives rise to a
pair of branches (stable and unstable) of the steady state bifurcation
diagram. These branches merge at the critical point, associated with the
respective potential minimum. It is also encouraging that the mean-field
Fokker-Planck equation is in remarkable agreement with the MC simulations of
the system dynamics.

There are also several new avenues of research. Indeed, we believe that the
theoretical-computational framework and associated methodologies presented
here can be useful for the study of bifurcations and phase transitions for
more complicated physical systems. Or indeed for systems where the potential
is known from experiments only, either physical or in-silico ones, and then
our framework can be adopted in a ``data-driven" approach. Of particular
interest would also be extension to multi-dimensional problems.
Two-dimensional problems in particular would be of direct relevance to
surface diffusion and therefore to technological processes in materials
science and catalysis. Other interesting extensions include additional
effects and complexities such as non-Markovian interaction particles, colored
and multiplicative noise and nonreversible
perturbations~\cite{LelievreNierPavliotis2013, DuncanLelievrePavliotis2016}.
Recall also from the Introduction, that our starting point, a system of
interacting particles in a confining potential, retains the main features of
DDFT models. Another interesting study would then be applications of our
framework to such models. Finally, the study of phase transitions for the
stochastic gradient descent dynamics algorithms and of their mean field limit
that are used in the training of neural networks is an intriguing problem,
with important potential applications. We shall examine these and related
questions in future studies.

\section*{Acknowledgements}
We are grateful to Ch. Kuehn for useful discussions, particularly about the method of moments.  We acknowledge financial support from the Engineering and Physical Sciences
Research Council (EPSRC) of the UK through Grants No. EP/L027186, EP/L020564/1, EP/K034154/1, EP/P031587/1, EP/L024926/1 and EP/L025159/1.

\appendix
\section{Formulas for $Z(m)$, $R(m)$ and $V(m)$, for piecewise linear potentials with quadratic growth}
\label{ap:1} Here we list the values of $Z(m)$, $R(m)$ and $V(m)$ for a
potential of the form in Section~\ref{sec:P-L}. We consider a general
potential with $2M$ minima which have heights $h_1,\dots, h_{2M}$ and are
located at $x_1,\dots x_{2M}$. There are, therefore, $2M-1$ maxima/barriers
$H_i$, $i=1,\dots,2M-1$, located at $y_i$. Throughout our study we take $y_i
= \frac{x_i+x_{i+1}}{2}$, but the formulas are valid in the general case. We
define
\begin{multline}
\alpha = \sqrt{\frac\beta{2\theta}}, \quad \gamma = \sqrt{\frac\beta{2(\theta+1)}}, \\
s_i = \frac{H_i-h_i}{x_i-x_{i+1}}, \quad  S_i = \frac{H_i-h_{i+1}}{x_i-x_{i+1}}, \\ \textrm{ and }  \quad f(\varsigma,x,s) = \erf(\varsigma(\theta(m-x) + s))
\end{multline}

Using these, we obtain
\begin{multline}\label{eq:ap_Z}
Z(m) = -\frac{\sqrt\pi}\beta\left\{\alpha{\rm e}^{(\alpha m \theta)^2}
\left[\sum_{i=1}^{2M-1}\left( f(\alpha, y_i,2s_i)-f(\alpha, x_i,2s_i)\right)
{\rm e}^{4\alpha^2\left(\frac{   \left(  h_i(y_i-m)+H_i(m-x_i)  \right)  \theta}{x_i-x_{i+1}}+  s_i^2\right)}\right.\right.\\
-\sum_{i=1}^{2M-1}\left( f(\alpha, y_i,-2S_i)-f(\alpha, x_{i+1},-2S_i)\right)
\left.
{\rm e}^{4\alpha^2\left(\frac { \left(h_{i+1}(m-y_i)+H_i(x_{i+1}-m)\right) \theta}{x_i-x_{i+1}}+ S_i^2\right)}\right]\\
\left.
+\gamma{\rm e}^{(\gamma m \theta)^2}
 \left(  \left(f\left(\gamma, x_1,-x_1\right)-1 \right)
{\rm e}^{\gamma^2 \left(  {x_1}^2-2 h_1 \right) (\theta+1) }
-{\rm e}^{\gamma^2  \left( x_{2M}^2-2 h_{2M} \right) (\theta+1)}
 \left( f\left(\gamma, x_{2M},-x_{2M}\right)+1 \right)  \right) \right\}
\end{multline}

{\footnotesize{
\begin{multline}\label{eq:ap_M}
R(m) =
\frac{ \left( {\rm e}^{\beta  \left( x_1 \left( m-\frac{x_1}{2}\right) \theta-h_1 \right)}-{\rm e}^{\beta  \left( x_{2M} \left( m-\frac{x_{2M}}{2} \right) \theta-h_{2M} \right) } \right)}{\theta(\theta+1)\beta Z(m)}+ \frac{\sqrt\pi}{\beta Z(m)}\times\\
\times\left[\frac {\alpha}{\theta}
\left\{\sum_{i=1}^{2M-1}  \left( \theta m-2S_i\right)
\left( f(\alpha,y_i,-2S_i) -f(\alpha,x_{i+1},-2S_i)\right)
{\rm e}^{- 4\alpha^2\left( \frac {  \left(h_{i+1}(y_i-m)+H_i(m-x_{i+1}) \right) \theta}{x_i-x_{i+1}} - S_i^2\right)}\right.
\right.\\
\left.
-\sum_{i=1}^{2M-1} \left(\theta m+2s_i\right) \left( f(\alpha,y_i,2s_i)-f(\alpha, x_i,2s_i)\right)
 {\rm e}^{4\alpha^2 \left(\frac { \left( h_i(y_i-m)+H_1(m-x_i) \right)  \theta}{x_i-x_{i+1}}+  s_i^2\right)}\right\}{\rm e}^{(\alpha m\theta)^2}\\
\left.+\frac {\gamma m\theta{\rm e}^{(\gamma m \theta)^2}}{\left( \theta+1 \right)}
\left\{\left( f\left(\gamma,x_1,-x_1\right)-1\right)
{\rm e}^{\gamma^2 \left( x_1^2-2 h_1 \right) (\theta+1) }-\left( f\left(\gamma,x_{2M},-x_{2M}\right)+1 \right)
{\rm e}^{\gamma^2\left( x_{2M}^2-2 h_{2M} \right) (\theta+1)} \right\}\right]
\end{multline}
}}

{\footnotesize{
\begin{multline}\label{eq:ap_V}
V(m) =
-{\rm e}^{\beta\, \left( x_1 \left( m-\frac{x_1}{2} \right) \theta-h_1 \right) }
\frac{(m+x_1)\theta+x_1}{(\theta+1)^2\beta Z(m)}
+{\rm e}^{ \beta\left( x_{2M} \left( m-\frac{x_{2M}}{2}\right) \theta-h_{2M} \right)}
\frac { (m+x_{2M})\theta+x_{2M} }{\beta \left( \theta+1 \right) ^2 Z(m)}\\
-\frac{\sqrt\pi\gamma\left( \beta\theta^2 m^2+\theta+1 \right) {\rm e}^{(\gamma m \theta)^2}}{(\theta+1)^2\beta^2 Z(m)}
\left[
\left(f\left(\gamma, x_1, -x_1\right) -1\right){\rm e}^{\gamma^2 \left( (x_1^2 -2h_1)(\theta+1) \right) }
-  \left( f\left(\gamma,x_{2M},-x_{2M}\right)+1 \right) {\rm e}^{\gamma^2 \left( (x_{2M}^2-2h_{2M})( \theta+1)\right) } \right]\\
+
\sum_{i=1}^{2M-1} \frac{\left((m+x_i)\theta + 2s_i\right){\rm e}^{\beta\left( x_i \left( m-\frac{x_i}{2} \right) \theta-h_i \right) }-
((m+x_{i+1})\theta-2S_i){\rm e}^{\beta \left( x_{i+1} \left( m-\frac{x_{i+1}}{2} \right) \theta-h_{i+1} \right) }
-2(s_i+S_i) {\rm e}^{\beta\left(  y_i  \left( m -\frac{y_i}{2}\right) \theta-H_i\right)}}
{\theta^2\beta Z(m)}
\\
-\frac{\sqrt\pi\alpha e^{(\alpha m\theta)^2}}{\theta^2\beta^2 Z(m)}
\left(
 \sum_{i=1}^{2M-1}\left( f(\alpha,y_i,2s_i) - f(\alpha.x_i,2s_i) \right)
\left( \beta\theta^2 m^2 + \theta + 4\beta s_i(m\theta+s_i) \right)
{\rm e}^{4\alpha^2 \left(\frac{\left(h_i(y_i-m) + H_i(m-x_i)\right)\theta}{x_i-x_{i+1}}+ s_i^2 \right)} \right.\\
\left.
-\sum_{i=1}^{2M-1} \left(f(\alpha,y_i,-2S_i)-f(\alpha,x_{i+1},-2S_i)\right)
 \left(\beta\theta^2 m^2 + \theta- 4 \beta S_i(\theta m-S_i) \right)
{\rm e}^{-4\alpha^2  \left( \frac{ \left(H_i(m-x_{i+1})+h_{i+1}(y_{i}-m)\right) \theta}{x_i-x_{i+1}}- S_i^2\right)}\right)
\end{multline}
}}

\bibliographystyle{plain}

\begin{thebibliography}{10}

\bibitem{Adams-etal_2018}
S.~Adams, N.~Dirr, M.~Peletier and J.~Zimmer.
\newblock Large deviations and gradient flows.
\newblock {\em P. Roy. Soc. A-Math. Phy.}, 371(2005):20120341, 2013.

\bibitem{ContAllgower}
E.~L. Allgower and K.~Georg.
\newblock {\em Introduction to Numerical Continuation Methods}.
\newblock Colorado State University, 1990.

\bibitem{Bain-Bartolo_2017}
N.~Bain and D.~Bartollo.
\newblock Critical mingling and universal correlations in model binary active
  liquids.
\newblock {\em Nat. Commun.}, 8:15969, 2017.

\bibitem{balescu97}
R.~Balescu.
\newblock {\em Statistical dynamics. Matter out of equilibrium}.
\newblock Imperial College Press, London, 1997.

\bibitem{Berthier-Biroli_2011}
L.~Berthier and G.~Biroli.
\newblock Theoretical perspective on the glass transition and amorphous
  materials.
\newblock {\em Rev. Mod. Phys.}, 83:587--645, 2011.

\bibitem{BinneyTremaine2008}
J.~Binney and S.~Tremaine.
\newblock {\em Galactic Dynamics}.
\newblock Princeton University Press, Princeton, second edition, 2008.

\bibitem{BKRS2015}
V.I.~Bogachev, N.V.~Krylov, M.~R\"ockner, and S.~V. Shaposhnikov.
\newblock {\em Fokker-{P}lanck-{K}olmogorov equations}, volume 207 of {\em
  Mathematical Surveys and Monographs}.
\newblock American Mathematical Society, Providence, RI, 2015.

\bibitem{Bonfanti-Kob_2017}
S.~Bonfanti and W.~Kob.
\newblock Methods to locate saddle points in complex landscapes.
\newblock {\em J. Chem. Phys.}, 2017

\bibitem{Burk-Knobloch_2006}
J.~Burk and E.~Knobloch.
\newblock Localized states in the generalized swift-hohenberg equation.
\newblock {\em Phys. Rev. E}, 73:056211, 2006.

\bibitem{Pavliotis_al_2018}
J.A.~Carrillo, R.S.~Gvalani, G.A.~Pavliotis and A. Schlichting.
\newblock Long-time behaviour and phase transitions for the McKean--Vlasov equation on the torus.
\newblock {\em arXiv:1806.01719}, 2018.

\bibitem{CMV2006}
J.A.~Carrillo, R.J.~McCann, and C.~Villani.
\newblock Contractions in the 2-{W}asserstein length space and thermalization
  of granular media.
\newblock {\em Arch. Ration. Mech. Anal.}, 179(2):217--263, 2006.

\bibitem{Chekmarev_2015}
S.F.~Chekmarev.
\newblock Protein folding as a complex reaction: a two-component potential for
  the driving force of folding and its variation with folding scenario.
\newblock {\em Plos One}, 10:0121640, 2015.

\bibitem{Dawson1983}
D.A.~Dawson.
\newblock Critical dynamics and fluctuations for a mean-field model of
  cooperative behavior.
\newblock {\em J. Statist. Phys.}, 31(1):29--85, 1983.

\bibitem{matcont}
A.~Dhooge, W.~Govaerts, Yu.A.~Kuznetsov, W.~Mestrom, A.M.~Riet, and   B.~Sautois.
\newblock {\em MATCONT and CL MATCONT: Continuation toolboxes in matlab}.
\newblock Utrecht University, Netherlands and Universiteit Gent, Belgium, 2006.

\bibitem{Dirr-etal_2016}
N.~Dirr, M.~Stamatakis and J.~Zimmer.
\newblock Entropic and gradient flow formulations for nonlinear diffusion.
\newblock {\em J. Math. Phys.}, 57(8):081505, 2016.

\bibitem{DuncanLelievrePavliotis2016}
A.B.~Duncan, T.~Leli{\`e}vre and G.A.~Pavliotis.
\newblock Variance {R}eduction {U}sing {N}onreversible {L}angevin {S}amplers.
\newblock {\em J. Stat. Phys.}, 163(3):457--491, 2016.

\bibitem{DuongPavliotis2018}
M.H.~Duong and G.A.~Pavliotis.
\newblock Mean field limits for non-Markovian interacting particles: convergence to equilibrium, GENERIC formalism, asymptotic limits and phase transitions.
\newblock {\em arXiv:1805.04959}, 2018.

\bibitem{Farkhooi2017}
F.~Farkhooi and W.~Stannat.
\newblock A complete mean-field theory for dynamics of binary recurrent neural
  networks.
\newblock {\em arXiv:1701.07128v1}, 2017.

\bibitem{frank04}
T.~D. Frank.
\newblock {\em Nonlinear {F}okker-{P}lanck equations}.
\newblock Springer Series in Synergetics. Springer-Verlag, Berlin, 2005.

\bibitem{GPY2012}
J.~Garnier, G.~Papanicolaou, and T.-W.~Yang.
\newblock Large deviations for a mean field model of systemic risk.
\newblock {\em SIAM Journal of Financial Mathematics}, {4}({1}):{151--184},
  2013.

\bibitem{GPY2017}
J.~Garnier, G.~Papanicolaou, and T.-W.Yang.
\newblock Consensus convergence with stochastic effects.
\newblock {\em Vietnam J. Math.}, 45(1-2):51--75, 2017.

\bibitem{Gartner1988}
J.~G\"{a}rtner.
\newblock On the {M}c{K}ean -{V}lasov limit for interacting diffusions.
\newblock {\em Math. Nachr.}, 137:197--248, 1988.

\bibitem{Ben_2012b}
B.D.~Goddard, A.~Nold, N.~Savva, G.A.~Pavliotis, and S.~Kalliadasis.
\newblock General dynamical density functional theory for classical fluids.
\newblock {\em Phys. Rev. Lett.}, 109:120603, 2012.

\bibitem{Ben_2013}
B.D.~Goddard, A.~Nold, N.~Savva, P.~Yatsyshin, and S.~Kalliadasis.
\newblock Unification of dynamic density functional theory for colloidal fluids
  to include inertia and hydrodynamic interactions: derivations and numerical
  experiments.
\newblock {\em J. Phys.: Condens. Matter}, 25:035101, 2013.

\bibitem{Ben_2012a}
B.D.~Goddard, G.A.~Pavliotis, and S.~Kalliadasis.
\newblock The overdamped limit of dynamic density functional theory: rigorous results.
\newblock {\em Multiscale Model. Simul.}, 10:633--663, 2012.

\bibitem{GomesPavliotis2017}
S.N.~Gomes and G.A.~Pavliotis.
\newblock Mean field limits for interacting diffusions in a two-scale potential.
\newblock {\em {J}ournal of {N}onlinear {S}cience}, 28(3):905--941,2018.

\bibitem{Kawai-Komatsuzaki_2010}
S.~Kawai and T.~Komatsuzaki.
\newblock Hierarchy of reaction dynamics in a thermally fluctuating
  environment.
\newblock {\em Phys. Chem. Chem. Phys.}, 12:7626--7635, 2010.

\bibitem{Keil_2012}
F.J.~Keil.
\newblock Multiscale modelling in computational heterogeneous catalysis.
\newblock In B.~Kirchner and J.~Vrabec, editors, {\em Multiscale Molecular
  Methods in Applied Chemistry}, volume 307, pages 69--107. Topics in Current
  Chemistry, Springer, 2012.

\bibitem{Krauskopf}
B.~Krauskopf.
\newblock {\em Numerical Continuation Methods for Dynamical Systems}.
\newblock Springer, 2007.

\bibitem{LelievreNierPavliotis2013}
T.~Lelievre, F.~Nier and G.A.~Pavliotis.
\newblock Optimal Non-reversible Linear Drift for the Convergence to Equilibrium of a Diffusion.
\newblock {\em J. Stat. Phys.}, 152(2):237--274, 2013.

\bibitem{Lucon2016}
E.~Lu\'{c}on and W.~Stannat.
\newblock Transition from gaussian to non-gaussian fluctuations for mean-field
  diffusions in spatial interaction.
\newblock {\em The Annals of Probability}, 26(6):3840--3909, 2016.

\bibitem{MartzelAslangul2001}
N.~Martzel and C.~Aslangul.
\newblock Mean-field treatment of the many-body {F}okker-{P}lanck equation.
\newblock {\em J. Phys. A}, 34(50):11225--11240, 2001.

\bibitem{mckean}
H.P.~McKean.
\newblock Propagation of chaos for a class of non-linear parabolic equations.
\newblock {\em Stochastic Differential Equations (Lecture Series in
  Differential Equations, Session 7, Catholic Univ., 1967)}, pages 41--57.,
  1967.

\bibitem{McKean1966}
H.P.~McKean, Jr.
\newblock A class of {M}arkov processes associated with nonlinear parabolic
  equations.
\newblock {\em Proc. Nat. Acad. Sci. U.S.A.}, 56:1907--1911, 1966.

\bibitem{Motsch2014}
S.~Motsch and E.~Tadmor.
\newblock Heterophilious dynamics enhances consensus.
\newblock {\em SIAM Review}, 56(4):577--621, 2014.

\bibitem{Muller-Brown_1979}
K.~M{\"u}ller and L.D.~Brown.
\newblock Location of saddle points and minimum energy paths by a constrained simplex
  optimization procedure.
\newblock {\em Theoret. Chim. Acta (Berl.)}, 53(1):75--93, 1979.

\bibitem{Oelschlager1984}
K.~Oelschl\"{a}ger.
\newblock A martingale approach to the law of large numbers for weak
  interacting stochastic processes.
\newblock {\em The Annals of Probability}, 12(2):458--479, 1984.

\bibitem{Pavl2014}
G.~A. Pavliotis.
\newblock {\em Stochastic processes and applications}, volume~60 of {\em Texts
  in Applied Mathematics}.
\newblock Springer, New York, 2014.
\newblock Diffusion processes, the Fokker-Planck and Langevin equations.

\bibitem{PavlSt08}
G.A.~Pavliotis and A.M.~Stuart.
\newblock {\em Multiscale Methods}, volume~53 of {\em Texts in Applied
  Mathematics}.
\newblock Springer, New York, 2008.
\newblock Averaging and Homogenization.

\bibitem{Pinnau_al2017}
R.~Pinnau, C.~Totzeck, O.~Tse, and S.~Martin.
\newblock A consensus-based model for global optimization and its mean-field
  limit.
\newblock {\em Math. Models Methods Appl. Sci.}, 27(1):183--204, 2017.

\bibitem{rotskoff_vanden-eijnden2018}
G.M.~Rotskoff and E.~Vanden-Eijnden.
\newblock Neural Networks as Interacting Particle Systems: Asymptotic Convexity of the Loss Landscape and Universal Scaling of the Approximation Error.
\newblock {\em arXiv:1805.00915}, 2018.

\bibitem{shiino1987}
M.~Shiino.
\newblock Dynamical behavior of stochastic systems of infinitely many coupled
  nonlinear oscillators exhibiting phase transitions of mean-field type: H
  theorem on asymptotic approach to equilibrium and critical slowing down of
  order-parameter fiuctuations.
\newblock {\em Physical Review A}, 36(5):2393--2412, 1987.

\bibitem{SirignanoSpiliopoulos2018}
J.~Sirignano and K.~Spiliopoulos.
\newblock Mean Field Analysis of Neural Networks.
\newblock {\em arXiv:1805.01053}, 2018.

\bibitem{Tamura1984}
Y.~Tamura.
\newblock On asymptotic behaviors of the solution of a non-linear diffusion
  equation.
\newblock {\em J. Fac. Sci. Univ. Tokyo}, 31:195--221, 1984.

\bibitem{Tugaut2014}
J.~Tugaut.
\newblock Phase transitions of {M}c{K}ean-{V}lasov processes in double-wells
  landscape.
\newblock {\em Stochastics}, 86(2):257--284, 2014.

\bibitem{Villani2003}
C.~Villani.
\newblock {\em Topics in optimal transportation}, volume~58 of {\em Graduate
  Studies in Mathematics}.
\newblock American Mathematical Society, Providence, RI, 2003.

\bibitem{Peter_2012}
P.~Yatsyshin, N.~Savva, and S.~Kalliadasis.
\newblock Spectral methods for the equations of classical density-functional
  theory: Relaxation dynamics of microscopic films.
\newblock {\em J. Chem. Phys.}, 136:124113, 2012.

\bibitem{Peter_2013}
P.~Yatsyshin, N.~Savva, and S.~Kalliadasis.
\newblock Geometry-induced phase transition in fluids: {C}apillary prewetting.
\newblock {\em Phys. Rev. E}, 87:020402({R}), 2013.

\bibitem{Peter_2016}
P.~Yatsyshin, A.O.~Parry, and S.~Kalliadasis.
\newblock Complete prewetting.
\newblock {\em J. Phys.: Condens. Matter}, 28:275001, 2016.

\bibitem{Peter_2018}
P.~Yatsyshin, A.O. ~Parry, C.~Rasc\'on  and S.~Kalliadasis.
\newblock Wetting of a plane with a narrow solvophobic stripe.
\newblock {\em Mol. Phys.}, 116:1990-1997, 2018.

\bibitem{zofiaThesis}
Z.~Trstanova.
\newblock {\em Mathematical and algorithmic analysis of modified Langevin
  dynamics}.
\newblock PhD thesis, Universit{\'e} Grenoble Alpes, 2016.

\end{thebibliography}

\end{document}